\documentclass[11pt]{article}

\usepackage[utf8]{inputenc}
\usepackage[T1]{fontenc}

\usepackage[a4paper,margin=1in]{geometry}

\usepackage{amsthm, amsmath, amssymb}
\usepackage{booktabs}
\usepackage{tabularx}
\usepackage{float}
\usepackage{xcolor}
\usepackage{enumitem}

\usepackage{eucal}

\usepackage{hyperref}
\hypersetup{%
  colorlinks=true,
  linkbordercolor=red,
  citecolor=green!70!black,
}

\usepackage{authblk} 
\usepackage{titlesec} 
 
\usepackage[linesnumbered,algoruled,boxed,lined]{algorithm2e}
\usepackage{complexity} 

\usepackage{needspace}
\usepackage{graphicx}
\usepackage{booktabs}
\usepackage{comment}
\usepackage{amsmath, amssymb, amsthm}
\usepackage{tikz}
\usetikzlibrary{shapes.multipart, positioning, arrows.meta, fit, patterns}
\usepackage{enumitem}
 
\usepackage{cleveref} 

\usepackage{soul}

\def\cqedsymbol{\ifmmode$\lrcorner$\else{\unskip\nobreak\hfil
\penalty50\hskip1em\null\nobreak\hfil$\lrcorner$
\parfillskip=0pt\finalhyphendemerits=0\endgraf}\fi}

\usepackage{tikz}
\usetikzlibrary{positioning,backgrounds,patterns,calc}

\usepackage{enumitem}
\usepackage{needspace}
\usepackage{thmtools}
\usepackage{thm-restate}

\newtheorem{theorem}{Theorem}[section]
\newtheorem{lemma}[theorem]{Lemma}

\newtheorem{definition}[theorem]{Definition}

\newtheorem{reduction rule}[theorem]{Reduction Rule}
\newtheorem{marking-scheme}{Marking Scheme}
\newtheorem{remark}[theorem]{Remark}

\usetikzlibrary{patterns}

\usepackage{subcaption}
\newcommand{\SoPP}{\textsc{Isometric Path Partition}\xspace}
\newcommand{\SoPC}{\textsc{Isometric Path Cover}\xspace}

\newcommand{\diam}{\textsf{diam}}
\newcommand{\tw}{\textsf{tw}}

\newcommand{\calP}{\mathcal{P}}
\newcommand{\calQ}{\mathcal{Q}}
\newcommand{\calS}{\mathcal{S}}

\makeatletter
\newcommand{\customlabel}[2]{%
\protected@write \@auxout {}{\string \newlabel {#1}{{#2}{}}}}
\makeatother

\newcommand{\defproblem}[3]{
  \vspace{1mm}
\noindent\fbox{
  \begin{minipage}{0.96\textwidth}
  \begin{tabular*}{\textwidth}{@{\extracolsep{\fill}}lr} #1 \\ \end{tabular*}
  {\bf{Input:}} #2 \\
  {\bf{Question:}} #3
  \end{minipage}
  }
  \vspace{1mm}
}

\newcommand{\defproblemout}[3]{
  \vspace{1mm}
\noindent\fbox{
  \begin{minipage}{0.96\textwidth}
  \begin{tabular*}{\textwidth}{@{\extracolsep{\fill}}lr} #1 \\ \end{tabular*}
  {\bf{Input:}} #2  \\
  {\bf{Output:}} #3
  \end{minipage}
  }
  \vspace{1mm}
}

\usepackage{ifthen}
\newboolean{long-version}
\setboolean{long-version}{true}
\newcommand{\iflongelse}[2]{\ifthenelse{\boolean{long-version}}{#1}{{\color{blue!80!white}#2}}}

\title{Parameterized Complexity of Isometric\\ Path Cover: Treewidth and Diameter}
	
\author[1]{Dibyayan Chakraborty} 
\author[2]{Oscar Defrain} 
\author[3]{Florent Foucaud} 

\author[4]{\authorcr Mathieu Mari} 
\author[5]{Prafullkumar Tale} 

\affil[1]{School of Computer Science, University of Leeds, UK}
\affil[2]{Aix-Marseille Université, CNRS, LIS, Marseille, France}
\affil[3]{Université Clermont Auvergne, CNRS, Mines Saint-Étienne,\authorcr \normalsize  Clermont Auvergne INP, LIMOS, 63000 Clermont-Ferrand, France}
\affil[4]{LIRMM, Université de Montpellier, CNRS, Montpellier, France}
\affil[5]{Indian Institute of Science Education and Research Pune, Pune, India}

\date{October 2026} 

\begin{document}

\maketitle

\begin{abstract}
We study the fixed-parameter tractability of the \textsc{Isometric Path Cover} problem where, given an integer $k$ and a simple, undirected, and unweighted graph $G$, the objective is to decide whether $G$ admits a set of $k$ isometric paths (i.e., shortest paths between their endpoints) whose union is the entire vertex set of $G$.

We show that \textsc{Isometric Path Cover} is NP-hard on graphs of pathwidth (and thus treewidth) at most 287. This is an interesting contrast with the closely related problems \textsc{Isometric Path Partition} and \textsc{Path Cover}, which are both known to be XP for treewidth, see [Chakraborty et al., WG 2026] and [Foucaud et al., TCS 2026].

On the positive side, we show, by a suitable dynamic programming scheme, that \textsc{Isometric Path Cover} admits an XP algorithm w.r.t.~the parameter combination of treewidth and diameter, running in time $n^{\diam^{O(\tw)}}$ where $\diam$, $\tw$, and $n$ denote the diameter, the treewidth, and the number of vertices of the input graph, respectively.

However, we also prove that \textsc{Isometric Path Cover} becomes W[1]-hard for pathwidth on graphs of diameter~7, showing that an XP algorithm is the best one can hope for w.r.t.~this parameterization, under standard complexity theory assumptions.

\vskip5pt\noindent{}{\bf Keywords:} \SoPC, parameterized complexity, parameterized reductions, treewidth, diameter.
\end{abstract}

\newpage
\tableofcontents
\newpage

\section{Introduction}\label{sec:intro}

Our goal is to study the parameterized complexity of \SoPC, where one aims at covering the vertices of an input graph using a minimum number of shortest (isometric) paths. A fixed-parameter-tractable (\FPT) running time has the form $f(k)n^{O(1)}$, and an \XP\ running time has the form $f(k)n^{g(k)}$~\cite{cygan2015parameterized}, where $n$ is the input size, $k$ the parameter, and $f,g$ computable functions. All graphs in this paper are unweighted, undirected, and simple. 

Covering graphs with paths with special properties is central to (algorithmic) graph theory. For example, decades of research have been devoted to the algorithmic and structural study of \emph{Hamiltonian graphs}, i.e., whose vertex set can be covered with only one path. \textsc{Hamiltonian Path}, i.e.,~the algorithmic task of deciding whether a graph is Hamiltonian, is a classic NP-complete problem~\cite{garey2002computers}, and the study of its fixed-parameter tractability gave rise to the famous \emph{cut-and-count} technique~\cite{cygan2022solving}, leading to single-exponential \FPT\ algorithms for many algorithmic problems parameterized by treewidth. A natural generalization of \textsc{Hamiltonian Path} is \textsc{Path Cover} where, given a graph $G$ and an integer $k$, the objective is to decide whether there exists a collection of $k$ paths that covers the vertex set of $G$. \textsc{Path Cover} derives motivation from the Gallai-Millgram Theorem~\cite{gallai1960verallgemeinerung}. The parameterized complexity of \textsc{Path Cover} has recently been studied~\cite{fomin2024hamiltonicity,FOUCAUD2026116226}. While \textsc{Path Cover} admits an \FPT\ algorithm w.r.t.~the combined parameters treewidth and solution size (via Courcelle's Theorem~\cite{DBLP:journals/iandc/Courcelle90}), it is a challenging open problem whether it is \FPT\ for treewidth alone~\cite{FOUCAUD2026116226}.

A path is \emph{isometric} if it is a shortest path between its endpoints. We consider a variant of \textsc{Path Cover} known as \SoPC, defined as follows.  

\medskip

\defproblemout{\SoPC}
{A graph $G$ and an integer $k$.}{Is there a collection of at most $k$ isometric paths in $G$, such that each vertex of $G$ belongs to at least one of the paths?}

\medskip

\SoPC was introduced in the context of the Cops and Robber game on graphs~\cite{FF01} and cop-decompositions~\cite{cop-decs}, that are tree-decomposition where the vertices of any bag can be covered by a bounded number of isometric paths in a suitably chosen subgraph. \SoPC serves a natural simplified model for network design problems of covering transportation networks by traffic routes. \SoPC has also been applied to machine learning on graphs~\cite{TG21}. Recently, the study of the approximability of \SoPC led to the discovery of the graph parameter \emph{isometric path complexity}~\cite{chakraborty2026isometric,chakraborty2025strong}, which connects graph hyperbolicity to aspects of coarse graph theory, an area that recently garnered considerable attention~\cite{georgakopoulos2025graph,nguyen2025coarse, hickingbotham2025graphs}. Moreover, graphs that can be covered with a bounded number of isometric paths have been shown to have bounded treewidth~\cite{baste2025polynomial,dumas2024graphs}; other results related to \SoPC are surveyed in Section~\ref{sec:related}.

\subsection{Our results}

Despite growing interest for \SoPC, a thorough study on its parameterized complexity is missing from the literature. The problem has been shown to be \FPT\ w.r.t.~parameters vertex cover and neighborhood diversity~\cite{PPP}.
Our first result shows that these results cannot be extended to pathwidth or treewidth, unless \P=\NP\ (even with an \XP\ algorithm). This negatively answers open questions from~\cite{foucaud2022,dumas2024graphs}.

\begin{restatable}{theorem}{restatethmIPCtwparaNPh}
	\label{thm:cover-tw-paraNPh}
	{\SoPC is \NP-complete for graphs
	of pathwidth (and thus treewidth) at most 287.}
\end{restatable}

The reduction is from an \NP-complete variant of perfect matching in cubic graphs which we call \textsc{Constrained Perfect Matching}:
the edges {of the input graph} are grouped into groups {of size~3}, and the matching must choose exactly one edge
from each group. Each vertex is incident with one selected edge and two
unselected edges.  Informally, our construction represents these three
choices by one primary slot and two dummy slots on the element side,
and one slot for each incidence $(S,u)$ on the incidence side. {Potential solution paths}
pair the incidence slots with the corresponding primary or dummy slots
of the same element.  Blocking vertices force the unselected sets to
use dummy paths.  The two sides communicate through a constant number
of gates; retaining these gates in the decomposition gives bounded
pathwidth.  

Theorem~\ref{thm:cover-tw-paraNPh} adds \SoPC to the growing list of distance-based graph optimization problems which are \NP-complete on graphs of constant pathwidth. 
Typically, in a distance-based graph optimization problem, the objective is to choose a solution sub-structure of the input graph satisfying some metric constraints while optimising the solution size. 
Some of these problems e.g., \textsc{(Strong) Metric Dimension}~\cite{LM21,tale2026hardnessstrongmetricdimension} and \textsc{Geodetic Set}~\cite{DBLP:conf/iwpec/Tale25}, are known to be \NP-complete on graphs of constant pathwidth, while others are W[1]-hard for pathwidth but \XP\ for treewidth, like \textsc{Isometric Path Partition}~\cite{ChakrabortyDFMT2026}, \textsc{Distance-$d$-Dominating Set}~\cite{KLP19}, \textsc{Distance-$d$-Independent Set}~\cite{DBLP:journals/dam/KatsikarelisLP22}.
To complement such hardness results, researchers have provided \FPT\ algorithms for distance-based optimization problems w.r.t.~cyclomatic number, vertex cover number, treedepth, the combination of treewidth and diameter etc.~\cite{ChakrabortyDFMT2026,DBLP:journals/jgaa/Eppstein15,DBLP:conf/icalp/FoucaudGK0IST24,FoucaudMDGS_VC,KK22,KLP19,DBLP:journals/dam/KatsikarelisLP22,DBLP:journals/algorithmica/KuceraS23}. 
The treewidth and diameter parameterization is particularly suitable for such problems, as it constrains both the structure and the distances in the graph. We prove the following.

\begin{restatable}{theorem}{restatethmIPCalgo}
	\label{thm:IPC:algo}
	\SoPC admits an algorithm running in time $n^{\diam^{O(\tw)}}$ where $\diam$, $\tw$, and $n$ denote the diameter, the treewidth, and the number of vertices of the input graph, respectively.
\end{restatable}

{We prove Theorem~\ref{thm:IPC:algo} by providing a dynamic programming algorithm. 
In this algorithm, the states describe the ``traces'' of paths on the bags of the decomposition, which consist of how paths meet a bag, which connections should be realized outside the bag, and some additional distance information of their endpoints. 
Together with these traces, we store multiplicities indicating how many paths from the solution yield a given trace.
Let us point that, although this approach appears to follow standard techniques, it requires special effort to achieve the claimed time bound.
Most notably, instead of storing the exact endpoints of the traces, we must limit ourselves to the \emph{distance profiles} of these endpoints, which are vectors indicating their distances to the vertices of a given bag of the decomposition.
Handling these require particular attention for the compatibility of states, introducing some complexity in proving the correctness of our algorithm.}

{In the light of the above result, a natural question is whether the dependence on treewidth and diameter can be confined to a multiplicative factor of the input size, that is, whether \SoPC admits an \FPT{} algorithm for this combination of parameters.} Under standard complexity theoretic assumptions, we rule out such a possibility by showing that \SoPC is \W[1]-hard when parameterized by the parameter combination of pathwidth and diameter. 

\begin{restatable}{theorem}{restatethmIPCWhtwdiam}
	\label{thm:cover-W1h-tw-diam}
	\SoPC is \W[1]-hard when parameterized by pathwidth
	even when the input graph has diameter $7$. 
\end{restatable}

The reduction is from a variant of \textsc{Unary Bin Packing}, where, given a positive integer $k$, item sizes and a bin capacity $b$ where the total item size is $kb$, all encoded in unary, the objective is to decide whether the
items fit into at most $k$ bins of capacity $b$.  

In our reduction, the constructed graph has (among others) $n$ vertices $a_1,\dots, a_n$ that correspond to the items, and $k$ vertices $b_1,\dots, b_k$ that correspond to the bins. Each vertex $a_i$ has leaves attached to it that encode the size of the corresponding item, and each vertex $b_j$ has $b$ leaves that encode the capacity of the corresponding bin. 
Then, each pair $(a_i,b_j)$ is connected by a structure called an \emph{anti-split} gadget that essentially consists of $t_i$ parallel paths plus some other edges and vertices. If the $i$-th item is placed into the $j$-th bin, then we create $t_i$ paths from the leaves of $a_i$ to $t_i$ distinct leaves of $b_j$ that traverse that gadget. A potential problem here is that, in order to cover the graph, one could be tempted to ``split'' the items by having some paths starting from the same item but finishing in different bins. Our anti-split gadget prevents such a split by inducing a ``penalty'' for each gadget that is only partially traversed.

\subsection{Related works.}\label{sec:related}

\SoPC is known to be NP-complete on chordal graphs~\cite{foucaud2022} and bipartite graphs with bounded degeneracy and diameter~\cite{PPP}. 
On the other hand, \SoPC admits polynomial time algorithms on restricted graph classes like chain graphs, cographs~\cite{chakraborty2026covering}, trees~\cite{FOUCAUD2026116226}, block graphs~\cite{DBLP:journals/ipl/PanC05}.  
\SoPC has also been studied from a structural point of view: the optimal solution sizes have been determined for square grids~\cite{FF01}, hypercubes~\cite{FNHC01}, complete $r$-partite graphs~\cite{PC06} and Cartesian products of complete graphs~\cite{PC06}. 
It was proved that if the vertex set of a graph can be covered with $k$ isometric paths, then its pathwidth is in $O(3^k)$~\cite{dumas2024graphs}; if also the edge set of the graph can be covered, then its pathwidth is  in $O(k^4)$~\cite{baste2025polynomial}. The former result was used in~\cite{dumas2024graphs} to obtain an \XP\ algorithm for the problem when parameterized by solution size.
Approximation algorithms for \SoPC were proposed, such as an $O(\log\diam(G))$-factor algorithm for general graphs~\cite{TG21} and constant factor approximation  algorithms for interval graphs, chordal graphs, and more generally on graphs with bounded hyperbolicity~\cite{chakraborty2026isometric,chakraborty2025strong,foucaud2022}. 

\paragraph*{Contrast with \SoPP.}

Chakraborty et al.~\cite{ChakrabortyDFMT2026} recently studied the \SoPP
problem through the lens of parameterized
complexity. {This problem has a similar definition as \SoPC, but the paths are required to be pairiwise vertex-disjoint.}
It was proved in~\cite{ChakrabortyDFMT2026} that \SoPP is \W[1]-hard w.r.t.~pathwidth but admits an \XP\ algorithm w.r.t.~treewidth. Moreover, \SoPP is fixed-parameter tractable w.r.t.~the combined parameter of treewidth and diameter. 
The results in this paper along with the ones in \cite{ChakrabortyDFMT2026} establish that \SoPC is significantly harder than \SoPP from the perspective of parameterized complexity. {These differences are summarized in Table~\ref{tab:results}. These are reminiscent with the (supposed) difference between \textsc{Path Cover} and \textsc{Path Partition} when parameterized by treewidth: for the former, only an \XP\ algorithm is known, while the latter is \FPT~\cite{FOUCAUD2026116226}.}

\begin{table}[ht]
\centering
\begin{tabularx}{\textwidth}{@{}
 >{\raggedright\arraybackslash}X
 >{\centering\arraybackslash}X
 >{\centering\arraybackslash}X
 @{}}
\toprule
parameter & \SoPP \cite{ChakrabortyDFMT2026} & \SoPC (this paper) \\
\midrule
\midrule
treewidth &  \XP\ for treewidth  & \NP-hard for pathwidth 287 \\
& \W[1]-hard for pathwidth & \\
\midrule
treewidth and diameter & \FPT\ & \XP\  \\
& & \W[1]-hard for pathwidth on diameter~7 graphs \\
\bottomrule
\end{tabularx}
\caption{Summary of our results on \SoPC\ and comparison with known results for \SoPP.}
\label{tab:results}
\end{table}

\subsection{Organisation of this paper}

In Section~\ref{sec:prelim}, we introduce some definitions used throughout the paper. In Sections~\ref{sec:SCP-tree-width-para-NP-hard},  \ref{sec:algo-treewidth-plus-diam}, and \ref{sec:cover-W1h-tw-diam} we prove Theorems~\ref{thm:cover-tw-paraNPh}, \ref{thm:IPC:algo}, and \ref{thm:cover-W1h-tw-diam}, respectively. We conclude in Section~\ref{sec:conclusion}.

\section{Preliminaries}
\label{sec:prelim}

A graph has vertex set $V(G)$ and edge set $E(G)$. We denote edge with endpoints $u, v$ as $(u, v)$. 
The \emph{length} of a path $P$ is the number of edges in $P$.  An \emph{isometric path} is a shortest path between its endpoints.
An \emph{IP-cover} of a graph $G$ is a collection
of isometric paths that covers
all the vertices in \(G\).
The \emph{size} of an IP-cover $\mathcal{P}$ of a graph $G$ is the cardinality of $\calP$. 
In the rest of the paper, we may denote an isometric path $P$ by the natural ordering $(x_1, x_2, \dots, x_{\ell})$ of its vertices obtained by traversing $P$ from one endpoint to the other.

For a graph $G$ and a set $X\subseteq E(G)$, the graph $G+X$ (resp.\ $G-X$) is the
graph obtained by adding (resp.\ removing) the edges in $X$ to (resp.\ from) $G$. 
For an induced subgraph $H$ of a graph $G$, and a set $X\subseteq V(G)$, $H+X$ (resp.\ $G-X$) is the 
subgraph of $G$ obtained by adding (resp.\ removing) the vertices in $X$ to (resp.\ from) $H$. 
When performing these operation, we ignore multiplicity, i.e., we keep a simple graph.
For a graph $G$ and a set $X\subseteq V(G)$, the graph $G[X]$ is the subgraph of $G$ induced by the vertices in $X$.
For convenience, we extend these notations to paths that we see as graphs, i.e., $P[X]$ for some path $P$ and $X\subseteq V(G)$ denotes the (possibly disconnected) graph consisting of subpaths of $P$ that are induced by $X$.

The notation of \emph{identifying} two vertices $u$ and $v$ in a graph $G$ is defined as follows:
The operation construct a simple graph $H$ by deleting vertices $u$ and $v$
from $G$, adding a new vertex $w$ and making it adjacent with
every vertex in $G$ that were adjacent with $u$ or $v$.

For a positive integer $q$, we denote {the} set $\{1, 2, \dots, q\}$ by $[q]$.

For sake of completeness we recall the well-known definitions of \emph{tree-decompositions}, 
\emph{treewidth} and \emph{nice tree-decompositions} below. 

\begin{definition}[\cite{cygan2015parameterized}]
A \emph{tree-decomposition} of a graph $G$ is a rooted tree $T$ 
where each node $v$ is associated to a subset $X_v$ of $V(G)$ called \emph{bag}, such that 
\begin{itemize}
\item The set of nodes of $T$ containing a given vertex of $G$ forms a nonempty connected subtree of $T$; and

\item Any two adjacent vertices of $G$ appear in a common node of $T$.

\end{itemize}
The \emph{width} of $T$ is the maximum cardinality of a bag minus one. The \emph{treewidth} of $G$ is the 
minimum integer $k$ such that $G$ has a tree-decomposition of width $k$.
\end{definition}

\begin{definition}[\cite{cygan2015parameterized}]
A \emph{nice tree-decomposition} of a graph $G$ is a rooted tree-decomposition such that each internal 
node has one or two children, with the following properties.
\begin{itemize}
\item Each node of $T$ belongs to one of the following types: \emph{introduce}, \emph{forget}, \emph{join} or \emph{leaf}.

\item A join node $v$ has two children $v_1$ and $v_2$ such that $X_v = X_{v_1} = X_{v_2}$.

\item An introduce node $v$ has one child $v_1$ such that $X_v \setminus \{x\} = X_{v_1}$, where $x \in X_v$.

\item A forget node $v$ has one child $v_{1}$ such that $X_v  = X_{v_1} \setminus \{x\}$, where $x \in X_{v_1}$.

\item A leaf node $v$ is a leaf of $T$ with $X_v=\emptyset$.

\item The tree $T$ is rooted at a node $r$ called root node with $X_r = \emptyset$.
\end{itemize}
\end{definition}

For a node $t$ in a nice tree-decomposition $T$ of a graph $G$, we let $G_t$ denote the subgraph of $G$ induced by the vertices in the union of bags of the nodes that belong to the subtree of $T$ rooted at $t$.

\providecommand{\CPM}{\textnormal{\textsc{Constrained Perfect Matching}}}
\providecommand{\IPC}{\textnormal{\textsc{Isometric Path Cover}}}
\providecommand{\cpmDistance}{\operatorname{dist}}
\providecommand{\Top}{\mathsf t}
\providecommand{\Bot}{\mathsf b}
\providecommand{\eps}{\varepsilon}
\providecommand{\figureplaceholder}[1]{%
  \fbox{\parbox{.90\linewidth}{\centering\small\itshape #1}}}
\providecolor{primarygreen}{HTML}{168A45}
\providecolor{dummyorange}{HTML}{D97706}

\section{\textsc{\SoPC} Hardness with Respect to Pathwidth}
\label{sec:SCP-tree-width-para-NP-hard}

In this section, we prove that \textsc{\SoPC} remains \NP-hard on graphs of pathwidth at most 287 (Theorem \ref{thm:cover-tw-paraNPh}). We first define the source problem and show its hardness, then we construct the graph of the reduction, and finally we analyze the reduction.

\subsection{The source problem: constrained perfect matching}
\label{subsec:cpm-source}

\begin{definition}[\CPM]
An instance consists of a finite universe $U$, a family $\mathcal S$ of
distinct two-element subsets of $U$, and a partition $\mathcal G$ of
$\mathcal S$ into \emph{groups}.  We impose the following restrictions:
\begin{enumerate}[label=(\roman*),nosep]
  \item every element belongs to exactly three sets of $\mathcal S$;
  \item every group contains exactly three sets;
  \item the three sets in a group are pairwise disjoint.
\end{enumerate}
The question is whether there is a subfamily $\mathcal F\subseteq\mathcal S$
such that $|\mathcal F\cap g|=1$ for every $g\in\mathcal G$ and
$\mathcal F$ is a partition of $U$.
\end{definition}

{\CPM\ can be reformulated as a constrained form of perfect matching in a cubic graph, since the sets in $\mathcal S$ can be  seen as edges, and where each group consists of three disjoint edges. Nevertheless, we prefer to use the above set formulation in order to avoid confusion with the constructed graph instance for \SoPC.}

Throughout the graph construction, members of $U$ are called
\emph{elements}, members of $\mathcal S$ are called \emph{sets}, and
members of $\mathcal G$ are called \emph{groups}.  Vertices and edges
refer to objects in the constructed graph, except for the auxiliary
source graph used in the following proof.

\begin{lemma}
\label{lem:cpm-source-hard}
{\CPM\ is \NP-complete.} 
\end{lemma}
\begin{proof}
Membership in \NP\ is immediate.  For hardness, reduce from proper
three-edge-colorability of simple cubic graphs, which is \NP-complete
by Holyer~\cite{Holyer1981}.

Given a simple cubic graph $H=(V,E)$, put $U=V\times\{1,2,3\}$.
For every $e=vw\in E$ and $j\in\{1,2,3\}$, introduce
\[
 S_{e,j}=\{(v,j),(w,j)\},
 \qquad
 g_e=\{S_{e,1},S_{e,2},S_{e,3}\}.
\]
Let $\mathcal S=\{S_{e,j}:e\in E,\ j\in\{1,2,3\}\}$ and
$\mathcal G=\{g_e:e\in E\}$.  Simplicity of $H$ ensures that all sets
are distinct.  Each $(v,j)$ belongs to precisely the three sets
corresponding to edges incident with $v$.  Each group consists of three
pairwise disjoint two-element sets.

If $\chi:E\to\{1,2,3\}$ is a proper edge-coloring, select
$S_{e,\chi(e)}$ for each $e$.  Exactly one edge of each color is incident
with each vertex of $H$, so every element $(v,j)$ is covered exactly once.
Conversely, a feasible subfamily selects a unique $S_{e,j}$ for every
$e$; give $e$ color $j$.  Exact coverage of $(v,j)$ says that precisely
one edge incident with $v$ has color $j$, so the coloring is proper.
The construction has polynomial size and is computable in polynomial time.
\end{proof}

Fix a nonempty instance $(U,\mathcal S,\mathcal G)$ and let $n=|U|$.
Counting incidences gives $|\mathcal S|=3n/2$ and $|\mathcal G|=n/2$;
in particular, $n$ is even.  The empty instance is handled separately.

\subsection{Construction of the graph}
\label{subsec:cpm-construction}

The graph has an element side, consisting of primary and dummy vertices
with leaves and a spine, and an incidence side, consisting of one set
gadget for each set, combined into group gadgets.  For each color of
elements, two interchanges connect these sides by element arms and
incidence arms.  Figure~\ref{fig:cpm-construction} illustrates the whole
construction. 

\begin{lemma}[Coloring the elements]
\label{lem:cpm-coloring}
In polynomial time, the elements can be colored with at most $16$ colors
so that the six elements occurring in any one group have distinct colors.
\end{lemma}
\begin{proof}
Form the conflict graph on $U$ in which two distinct elements are adjacent
if they occur in the union of some group.  An element belongs to three
sets and, because sets in a group are disjoint, to three different groups.
In each group it has at most five other elements as neighbors.  Therefore
the conflict graph has maximum degree at most $15$.  A greedy coloring
uses at most $16$ colors.  Both the graph and its coloring can be computed
in polynomial time.
\end{proof}

For each element $u$, let $c(u)$ be its color.  We use only nonempty
color classes and write $n_c$ for the number of elements of color $c$.

Independently, choose an arbitrary bijection $i:U\to\{1,\ldots,n\}$.
Set
\begin{equation}
 B=n+2,\qquad L=10((4n-1)B+n+1),\qquad k=12n+1.
 \label{eq:cpm-parameters}
\end{equation}

We use a row index $\alpha\in\{\Top,\Bot\}$ and an interchange sign
$\beta\in\{+,-\}$.  A bar reverses a row or a sign.  Write
$\eps_+=1$ and $\eps_-=-1$.  

\paragraph{Interchanges.}
For each color $c$, construct two disjoint, identically structured
interchanges $X_c^+$ and $X_c^-$.  Interchange $X_c^\beta$ has six gates:
\[
 g_{c,\Top}^{\beta},g_{c,\Bot}^{\beta},\qquad
 p_{c,\Top}^{\beta},p_{c,\Bot}^{\beta},\qquad
 d_{c,\Top}^{\beta},d_{c,\Bot}^{\beta}.
\]
These are, respectively, its group gates, primary element gates, and
dummy element gates.  Join the two group gates by an edge, and make the
four element gates a clique.  There are no edges directly joining a
group gate to an element gate.

The interchange also contains three mutually disjoint sets of internal
vertices.  Their union is independent, and their neighborhoods are:
\begin{center}
\begin{tabular}{lll}
\toprule
Internal set & Size & Neighborhood of each vertex\\
\midrule
$R_c^\beta$ & $2n_c$ &
$\{g_{c,\Top}^\beta,g_{c,\Bot}^\beta,
    p_{c,\Top}^\beta,p_{c,\Bot}^\beta\}$\\
$D_{c,\Top}^\beta$ & $2n_c$ &
$\{g_{c,\Top}^\beta,d_{c,\Top}^\beta\}$\\
$D_{c,\Bot}^\beta$ & $2n_c$ &
$\{g_{c,\Bot}^\beta,d_{c,\Bot}^\beta\}$\\
\bottomrule
\end{tabular}
\end{center}
Thus there are $6n_c$ internal vertices per interchange.

\paragraph{Element vertices and element leaves.}
For every $u\in U$, introduce four vertices
\[
 p_u^\Top,p_u^\Bot,\qquad d_u^\Top,d_u^\Bot.
\]
Attach two private leaves to each primary vertex $p_u^\alpha$, and
four private leaves to each dummy vertex $d_u^\alpha$.
These twelve leaves are the \emph{element leaves} of $u$.
The element vertices carry a row index but \emph{no sign index}:
each is connected to both signed interchanges of its color.

\paragraph{Element arms.}
For each $u\in U$ and each $(\alpha,\beta)$, add one primary element arm
from $p_u^\alpha$ to $p_{c(u),\alpha}^\beta$ and two dummy element arms
from $d_u^\alpha$ to $d_{c(u),\alpha}^\beta$.  The two dummy arms have
the same endpoints but disjoint internal vertices; they are distinct
paths, not parallel edges.  Every one of these arms has
\begin{equation}
 L+\eps_\beta i(u)\ \text{internal vertices, and edge length}\quad
 L+\eps_\beta i(u)+1.
 \label{eq:cpm-element-arm-length}
\end{equation}
An arm incident with $X_c^\beta$ on the element side has sign $\beta$.
There are four primary and eight dummy element arms for $u$, with six
top arms and six bottom arms.

\paragraph{The spine.}
Write $u_1,\ldots,u_n$ for the elements in increasing index order.
List the $4n$ element vertices as
\[
 p_{u_1}^{\Top},p_{u_1}^{\Bot},d_{u_1}^{\Top},d_{u_1}^{\Bot},
 \ldots,
 p_{u_n}^{\Top},p_{u_n}^{\Bot},d_{u_n}^{\Top},d_{u_n}^{\Bot}.
\]
Join each consecutive pair by a path of edge length $B$, with $B-1$
new internal vertices.  All these interiors are mutually disjoint
and disjoint from the arms and gadgets.  Their union is the
\emph{spine}, of length $(4n-1)B$.  Attach a new leaf $s_\star$ to its first
vertex and a new leaf $t_\star$ to its last vertex.  These are the two
\emph{auxiliary leaves}.  Let $Q_{\mathrm{sp}}$ be the path from
$s_\star$ to $t_\star$ along the spine, of length $(4n-1)B+2$.
The element leaves and the auxiliary leaves are collectively
called \emph{spine-side leaves}; the spine together with all these leaves
is the \emph{spine component} when arm interiors are deleted.
Spine segments are not called arms.

\paragraph{Set gadgets.}
Choose an ordering $S=(u,v)$ of the two elements of every set $S$.
The set gadget $H_S$ contains eight vertex-disjoint paths
\[
 P_{S,j}=f_{S,j}\,m_{S,j}\,\ell_{S,j},\qquad j\in\{1,\ldots,8\},
\]
each of edge length two.  The vertices $f_{S,j}$, $m_{S,j}$, and
$\ell_{S,j}$ are its first, middle, and last vertices.  The row labels are
fixed by the following table:
\begin{center}
\begin{tabular}{c*{8}{c}}
\toprule
$j$ & 1&2&3&4&5&6&7&8\\
\midrule
element $u(S,j)$ & $u$&$u$&$u$&$u$&$v$&$v$&$v$&$v$\\
row $\alpha(j)$ & $\Top$&$\Bot$&$\Top$&$\Bot$&$\Top$&$\Bot$&$\Top$&$\Bot$\\
interchange $\beta(j)$ & $+$&$+$&$-$&$-$&$+$&$+$&$-$&$-$\\
\bottomrule
\end{tabular}
\end{center}
Thus odd-numbered paths are top paths and even-numbered paths are bottom
paths.  Introduce gathering vertices $h_S^\Top,h_S^\Bot$ and attach four
private leaves to each.  Join $h_S^\alpha$ to $f_{S,j}$ whenever
$\alpha(j)=\alpha$.  The eight new leaves are the \emph{group leaves}
belonging to $S$.

\paragraph{Group gadgets.}
For $g=\{S_1,S_2,S_3\}$, form the group gadget $H_g$ from
$H_{S_1},H_{S_2},H_{S_3}$ and the following vertices and edges.
Introduce sixteen blocking vertices
\[
 y_{g,1},\ldots,y_{g,8},\qquad z_{g,1},\ldots,z_{g,8}.
\]
Let $j^+=j+1$ for $j<8$ and $8^+=1$; likewise let $j^-$ denote the
cyclic predecessor.  For every $j$ and every $S\in g$, add
\[
 y_{g,j}f_{S,j},\quad z_{g,j}f_{S,j},\qquad
 y_{g,j}m_{S,j^+},\quad z_{g,j}m_{S,j^+}.
\]
Finally, for every ordered pair of distinct sets $S,T\in g$, add all
edges
\begin{equation}
 f_{S,j}\ell_{T,r},\qquad j,r\in\{1,\ldots,8\}.
 \label{eq:cpm-cross-bicliques}
\end{equation}
Thus $H_g$ has $94$ nonleaf vertices and $24$ group leaves.

\paragraph{Incidence arms.}
For each $S\in\mathcal S$ and row $j$, put
$u=u(S,j)$, $c=c(u)$, $\alpha=\alpha(j)$, and $\beta=\beta(j)$.
Connect $\ell_{S,j}$ to $g_{c,\alpha}^{\beta}$ by an incidence arm
$I_{S,j}$ with
\begin{equation}
 L-\eps_\beta i(u)\ \text{internal vertices, and edge length}\quad
 L-\eps_\beta i(u)+1.
 \label{eq:cpm-incidence-arm-length}
\end{equation}
Its sign is $\bar\beta$, the \emph{opposite} of its interchange sign.
Each occurrence $(S,u)$ has four incidence arms, one for each
$(\alpha,\beta)$.  Since $u$ belongs to three sets, it has twelve
incidence arms in total.

All arms have mutually disjoint internal vertices, and these vertices
have no other neighbors.  No edges are added beyond those specified
above.

For reference, the sign convention is summarized by
\begin{center}
\begin{tabular}{lcc}
\toprule
Arm & Interchange & Number of internal vertices\\
\midrule
positive element arm & $X_c^+$ & $L+i(u)$\\
negative element arm & $X_c^-$ & $L-i(u)$\\
positive incidence arm & $X_c^-$ & $L+i(u)$\\
negative incidence arm & $X_c^+$ & $L-i(u)$\\
\bottomrule
\end{tabular}
\end{center}
In particular, the two arms of a route through $X_c^\beta$, from an
occurrence of $u$ to an element vertex of $v$, contribute
\begin{equation}
 (L-\eps_\beta i(u)+1)+(L+\eps_\beta i(v)+1)
   =2L+2+\eps_\beta\bigl(i(v)-i(u)\bigr).
 \label{eq:cpm-arm-cancellation}
\end{equation}
This is $2L+2$ when $v=u$.

This completes the simple unweighted graph $G$.  The output instance is
$(G,k)$ with $k=12n+1$.  

\begin{figure}[htbp]
\centering
\includegraphics[width=\linewidth]{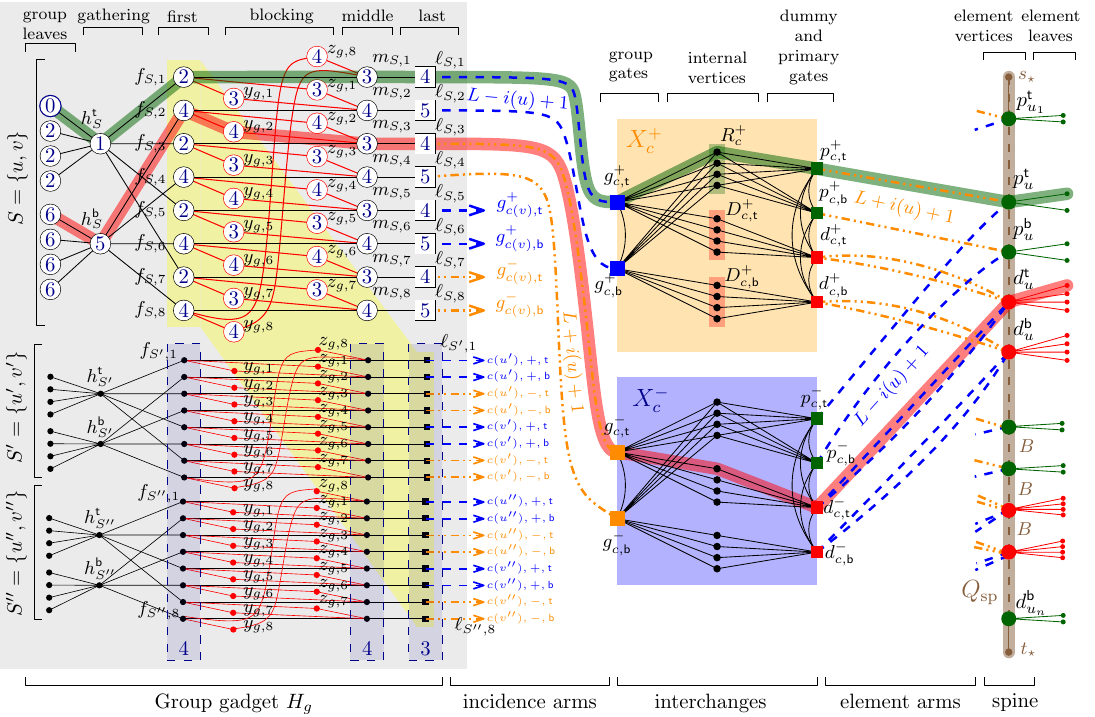}
\caption{The group gadget $H_g$ for $g=\{S,S',S''\}$, its connections to
 the interchanges, and the element side with the spine $Q_{\mathrm{sp}}$.
 Only the interchanges for $c=c(u)$ are detailed; the other incidence
 arms are labeled by their destinations.  Blue dashed arms are negative
 and orange dotted arms are positive.  The yellow part represents the biclique between the first vertices of the first set, and the last vertices of the two other sets. The two other bicliques are not shown. The highlighted green and red
 paths respectively illustrate a primary and a dummy path emerging from $H_S$. Blue numbers give distances
 from the group leaf marked $0$.  The brown spine joins
 $s_\star$ to $t_\star$ through the element vertices. }
\label{fig:cpm-construction}
\end{figure}

\subsection{Pathwidth and size of the construction}
\label{subsec:cpm-structure}

\begin{lemma}[Pathwidth and construction time]
\label{lem:cpm-pathwidth}
\label{lem:cpm-size}
The graph $G$ has pathwidth at most $287$ and can be constructed in
polynomial time.  It has $O(n^3)$ vertices and edges.
\end{lemma}
\begin{proof}
Keep all gates in every bag; there are at most $16\cdot2\cdot6=192$.
Process internal interchange vertices individually.  Then scan the
spine in order, retaining each element vertex while processing its
leaves and scanning its arms with two consecutive arm vertices.
These steps need at most three vertices in addition to the gates.

Finally, process each group gadget $H_g$ separately.  Retain its $94$
nonleaf vertices, introduce its leaves one at a time, and scan each
incidence arm using two consecutive arm vertices.  The largest bag
has at most $192+94+2=288$ vertices.  Every edge is covered and each
vertex occurs in an interval of bags, proving the pathwidth bound.

There are $24n$ arms, each of length $O(L)=O(n^2)$.  The group gadgets,
interchange vertices, and element vertices contribute $O(n)$ vertices
and edges, and the spine contributes $O(nB)=O(n^2)$.  The construction
therefore has $O(n^3)$ vertices and edges.  All steps, including the
coloring, take polynomial time.
\end{proof}

\subsection{Constructing a cover from a constrained perfect matching}
\label{subsec:cpm-forward}

\paragraph{Primary and dummy paths.}
For an occurrence $(S,u)$, choose a row $j$ for $u$ and put
$\alpha=\alpha(j)$, $\beta=\beta(j)$, and $c=c(u)$.
An \emph{$(S,u)$-primary path} starts at a leaf of $h_S^\alpha$, follows
\[
 h_S^\alpha,f_{S,j},m_{S,j},\ell_{S,j},
\]
then follows $I_{S,j}$ to $g_{c,\alpha}^\beta$, a vertex of $R_c^\beta$
to $p_{c,\delta}^\beta$, and the primary element arm to $p_u^\delta$,
where $\delta\in\{\Top,\Bot\}$.  It ends at an element leaf
of $p_u^\delta$.  The forward construction uses $\delta=\alpha$. An example is shown in green in Figure \ref{fig:cpm-construction}. 

An \emph{$(S,u)$-dummy path} starts at a leaf of $h_S^{\bar\alpha}$,
follows
\[
 h_S^{\bar\alpha},f_{S,j^-},w,m_{S,j},\ell_{S,j},
 \qquad w\in\{y_{g,j^-},z_{g,j^-}\},
\]
where $g$ contains $S$, then follows $I_{S,j}$ to $g_{c,\alpha}^\beta$,
a vertex of $D_{c,\alpha}^\beta$ to $d_{c,\alpha}^\beta$, and a dummy
element arm to $d_u^\alpha$.  It ends at an element leaf there.
The label $u$ refers to the exit row $j$: at the transitions $4\to5$
and $8\to1$, the preceding first vertex belongs to the other element
of $S$.  The detour nevertheless stays in the same set gadget.  An example is shown in red in Figure \ref{fig:cpm-construction}. 

\begin{remark}
Primary paths avoid the blocking vertices; each dummy path covers
exactly one.  This distinction enforces the choice of one set per group.
\end{remark}

\begin{lemma}[Distances from a group leaf]
\label{lem:cpm-local-distances}
Let $S,T\in g$ and let $\lambda$ be a leaf of $h_S^\alpha$.  Then
\begin{equation}
 \cpmDistance_G(\lambda,\ell_{T,j})=
 \begin{cases}
  3,&T\ne S,\\
  4,&T=S\text{ and }\alpha(j)=\alpha,\\
  5,&T=S\text{ and }\alpha(j)\ne\alpha.
 \end{cases}
 \label{eq:cpm-group-distances}
\end{equation}
\end{lemma}
\begin{proof}
Cross edges reach a last vertex of another set in three steps; a direct
row reaches a last vertex of $S$ with matching parity in four; and the
blocking detour reaches the opposite parity in five.  These routes are
shortest in $H_g$: the vertices reached in three steps are middle
vertices of the starting parity, blockers, and last vertices of other
sets, none adjacent to an opposite-parity last vertex of $S$.
A path leaving $H_g$ and returning uses at least two arms, of total
length at least $2(L-n+1)>5$, so it cannot improve these distances.
\end{proof}

\begin{lemma}[The prescribed paths are isometric]
\label{lem:cpm-isometric}
Every primary path has length $2L+9$, and every dummy path has length
$2L+10$.  These paths and $Q_{\mathrm{sp}}$ are isometric in $G$.
\end{lemma}
\begin{proof}
The two arms have total length $2L+2$.  A primary path additionally
uses four edges from its group leaf to the last vertex, two inside
the interchange, and one final leaf edge.  A dummy path uses five,
two, and one such edges, respectively.

A competing shortest path uses an even number of arms, and at least
two.  Four arms already cost at least $4(L-n+1)>2L+10$, so it uses
exactly two.  If it reaches the spine at an element vertex different from
the target leaf's neighbor, it traverses at least one spine segment of
length $B$.  Its arms have total length at least
$2L+2-(n-1)=2L+3-n$ by~\eqref{eq:cpm-arm-cancellation}.  Its group
prefix has length at least three, its interchange segment at least
two, and its final leaf edge has length one.  Consequently its length
is at least
\[
 2L+9-n+B=2L+11,
\]
which is too long.  Thus a shortest competitor uses no spine edge and
its element arm is incident with the prescribed target element vertex.

The interchange has color $c(u)$.  Since $u$ is the unique element of
that color in the source group, its incidence arm belongs to $(S,u)$.
For a primary endpoint, the minimum group-prefix and interchange
lengths are four and two.  For a dummy endpoint of parity $\alpha$
starting at $h_S^{\bar\alpha}$, exiting through parity $\alpha$ costs
$5+2$ edges, while exiting through parity $\bar\alpha$ costs $4+3$.
The group distances follow from Lemma~\ref{lem:cpm-local-distances};
the gate distances follow directly from the interchange.  Either sign
gives arm sum $2L+2$.
The stated path lengths therefore attain the endpoint distances.

Finally, $Q_{\mathrm{sp}}$ has length $(4n-1)B+2$.  Any competing path
that leaves the spine uses at least two arms, of total length at least
$2(L-n+1)>(4n-1)B+2$.  Thus $Q_{\mathrm{sp}}$ is isometric as well.
\end{proof}

\paragraph{The cover.}
Let $\mathcal F$ be a constrained perfect matching.  For every occurrence
$(S,u)$ and each of its four rows, take a primary path if $S\in\mathcal F$
and a dummy path otherwise.  Primary paths use the element gate of
their exit row.  In each group, use the $y$-blockers for one unselected
set and the $z$-blockers for the other.  Assign leaves, parallel dummy
arms, and internal interchange vertices bijectively to the paths using
them, and add $Q_{\mathrm{sp}}$.

\begin{lemma}[Forward implication]
\label{lem:cpm-forward}
If the source instance has a constrained perfect matching, then $G$
has an isometric path cover of size $k=12n+1$.
\end{lemma}
\begin{proof}
Each element has one selected and two unselected occurrences, giving
four primary and eight dummy paths.  Thus the construction uses $12n$
paths plus $Q_{\mathrm{sp}}$, all isometric by
Lemma~\ref{lem:cpm-isometric}.

The bijections are possible because each interchange has $2n_c$
primary paths and $2n_c$ dummy paths of each row, matching the sizes of
its internal vertex classes; each element vertex likewise has one
path per arm and leaf.  In a selected set every row is covered directly.
In an unselected set the path exiting row $j$ covers its middle and
last vertices, while the path exiting row $j^+$ covers its first
vertex.  The two unselected sets cover the sixteen blockers using
$y$ and $z$, respectively.  Every incidence arm is used once, and
$Q_{\mathrm{sp}}$ covers the spine and its auxiliary leaves.  This
accounts for all vertices.
\end{proof}

\subsection{Recovering a constrained perfect matching from a cover}
\label{subsec:cpm-reverse}

\begin{definition}[Nice paths and covers]
A \emph{nice path} is an isometric path from a group leaf to a spine-side
leaf (an element leaf or one of $s_\star,t_\star$), using exactly one
incidence arm, one internal interchange vertex, and one element arm.
A \emph{nice cover} consists of $Q_{\mathrm{sp}}$ and $12n$ nice paths
covering $V(G)$.  This definition imposes no primary or dummy form
inside the group gadget and no agreement of element indices.
\end{definition}

\begin{lemma}[Existence of a nice cover]
\label{lem:cpm-nice-cover}
\label{lem:cpm-spine-cut}
If $G$ has an isometric path cover of size at most $k$, it has a nice
cover.  The cover can be chosen so that its nice paths use no spine
edges and cover every arm and internal interchange vertex exactly once.
\end{lemma}
\begin{proof}
There are $24n+2=2k$ leaves, so every path in such a cover has two leaf
endpoints, each used once.  Arm interiors have degree two, so an arm
is always traversed completely.

We first bound arm usage.  Two spine-side leaves have distance at most
$(4n-1)B+2$.  From a group leaf, a direct row and two matching primary
arms reach the spine in $2L+8$ steps.  Continuing along the spine gives
distance bounds $2L+(4n-1)B+9$ to a spine-side leaf and
$4L+(4n-1)B+16$ to a group leaf.  Arm usage is even, and each arm has
length at least $L-n+1$.  Our choice of $L$ therefore gives:
\begin{center}
\begin{tabular}{lcc}
\toprule
Group endpoints & Arms used & Internal interchange vertices\\
\midrule
0 & 0 & 0\\
1 & 2 & at most 1\\
2 & at most 4 & at most 2\\
\bottomrule
\end{tabular}
\end{center}
For the last column, observe that a shortest path cannot leave and
re-enter an interchange: the gates have distance at most three, whereas
an excursion uses at least two arms.  A visit consequently has length
at most three and contains at most one internal vertex; it uses two arms.

There are $12n$ internal interchange vertices to cover and $12n$ group
endpoints.  Equality is forced in the last column for every path.
All internal vertices are therefore covered once, and the total arm
usage is $24n$.  Since all $24n$ arms must be covered, each is used once.
An interchange visit containing an internal vertex crosses between a
group gate and an element gate; using two gates on the same side would
have a clique shortcut.  Hence each group endpoint leads through two
arms and one internal vertex to the spine.  The bounds leave only a
spine portion between these branches, with no intervening group visit.

Cut the $12n$ branches at their first element vertices.  Two arrive
at each primary vertex and four at each dummy vertex, exactly matching
the numbers of element leaves.  Append these leaves bijectively and
add $Q_{\mathrm{sp}}$.  Each branch remains shortest, and extending it
to a leaf preserves shortestness.  The branches retain coverage away
from the spine, while the reassigned leaves and $Q_{\mathrm{sp}}$ cover
the rest.  This gives the asserted nice cover.
\end{proof}

For the rest of the proof, fix a nice cover with the additional properties
of Lemma~\ref{lem:cpm-nice-cover}.  Its nice paths end at element leaves
and visit only the group gadget containing their group endpoint.

\begin{lemma}[The group prefixes]
\label{lem:cpm-prefixes}
Each nice path has one of the following prefixes, for its endpoint
leaf $\lambda$ at $h_S^\alpha$:
\begin{align*}
 &\lambda,h_S^\alpha,f_{S,j},m_{S,j},\ell_{S,j}
       &&\text{if }\alpha(j)=\alpha;\\
 &\lambda,h_S^\alpha,f_{S,j^-},w,m_{S,j},\ell_{S,j}
       &&\text{if }\alpha(j)=\bar\alpha,
\end{align*}
where $w\in\{y_{g,j^-},z_{g,j^-}\}$ and $g$ contains $S$.
Call these primary and dummy prefixes, respectively.
Every first, middle, and last row vertex is covered exactly once.
\end{lemma}
\begin{proof}
By Lemma~\ref{lem:cpm-local-distances}, the prefix to the exit last
vertex has length at most five.  If it contains a middle vertex, that
vertex is reached in three or four steps.  In three steps it is reached
through its own first vertex and must exit through its own last vertex.
In four steps it must be reached through a blocker: entering from its
last vertex would require repeating that vertex, and no first vertex
has distance three from the starting leaf.  Exiting through another
set would have distance three and thus cannot use a middle vertex.
These are precisely the displayed prefixes, each with one middle vertex.

There are $12n$ middle vertices and $12n$ nice paths, each covering at
most one.  Coverage forces each path to have one of these prefixes and
each middle vertex to be covered once.  In each group, its $24$ paths
then cover one first and one last vertex each.  The $24$ vertices of
each kind must all be covered, so their coverage is also unique.
\end{proof}

\begin{lemma}[Path types and group selection]
\label{lem:cpm-selected-set}
Each group has eight paths with primary prefixes and sixteen with
dummy prefixes.  Primary prefixes lead to primary element vertices,
and dummy prefixes lead to dummy element vertices.  Exactly one set of each group has eight primary prefixes,
and each other set has eight dummy prefixes.  Each blocking vertex
is covered exactly once.
\end{lemma}
\begin{proof}
Only dummy prefixes cover blocking vertices, each covering one.
The sixteen blockers in a group require at least sixteen dummy prefixes,
leaving at most eight primary prefixes.  Globally there are thus at
most $8|\mathcal G|=4n$ primary prefixes.

A dummy prefix cannot lead to a primary element endpoint.  Switch
its exit to the other parity for the same element and interchange
sign.  The prefix length decreases from five to four, while either
group gate has distance two to the target primary gate.  The arm
lengths are unchanged, giving a strictly shorter route.  Since there
are $4n$ primary element arms to cover, and each nice path uses exactly
one element arm, at least $4n$ prefixes are primary.  Equality holds throughout: every
group has eight primary and sixteen dummy prefixes, primary prefixes
lead to primary endpoints, and dummy prefixes lead to dummy endpoints.
Each blocker is consequently covered once.

For a fixed set $S$, let $p_j,d_j$ indicate whether its unique exit
through $\ell_{S,j}$ has a primary or dummy prefix.  Unique coverage
of last vertices gives $p_j+d_j=1$.  Unique coverage of $f_{S,j}$ gives
$p_j+d_{j^+}=1$, since a dummy prefix exiting row $j^+$ visits $f_{S,j}$.
Thus $d_j=d_{j^+}$ for every cyclic index.  All eight prefixes of $S$
have the same type.  The count of eight primary prefixes per group
therefore selects exactly one of its three sets.
\end{proof}

\begin{lemma}[Index agreement by summation]
\label{lem:cpm-agreement}
Every nice path whose incidence arm belongs to $(S,u)$ has its element
endpoint at $u$.
\end{lemma}
\begin{proof}
Let $P$ be a nice path whose incidence arm belongs to $(S,u)$, and
write $v$ for its endpoint element and $\beta$ for its interchange sign.
Its two arms meet the same interchange, so $c(u)=c(v)$.

A route through the opposite sign exists with the same endpoint leaves
and the same local costs: choose the row for $u$ with the same exit
parity and opposite sign, and copy the interchange segment in the
identical opposite-sign interchange.  The group prefix has the same
length by Lemma~\ref{lem:cpm-prefixes}; for a dummy prefix, use the
corresponding predecessor blocker.  The total
arm length changes from $2L+2+\eps_\beta(i(v)-i(u))$ to
$2L+2-\eps_\beta(i(v)-i(u))$.  Shortestness therefore gives
\begin{equation}
 \eps_\beta\bigl(i(v)-i(u)\bigr)\le0.
 \label{eq:cpm-admissible-primary}
\end{equation}

Fix $\beta$.  Every element $u$ labels exactly six incidence arms
entering interchanges of sign $\beta$, and exactly six element arms
incident with those interchanges.  Every arm is used once by
Lemma~\ref{lem:cpm-spine-cut}.  Hence, summing over the nice paths
through sign $\beta$, with $u_P$ and $v_P$ denoting the incidence and
endpoint elements of $P$, respectively, we obtain
\begin{equation}
 \sum_P \eps_\beta\bigl(i(v_P)-i(u_P)\bigr)
 =\eps_\beta\left(6\sum_{v\in U}i(v)-6\sum_{u\in U}i(u)\right)=0.
 \label{eq:cpm-index-balance}
\end{equation}
Every summand is nonpositive by~\eqref{eq:cpm-admissible-primary},
so every summand is zero.  Since $i$ is injective, $u_P=v_P$ for every
path through either sign.
\end{proof}

\begin{lemma}[Reverse implication]
\label{lem:cpm-reverse}
If $G$ has an isometric path cover of size at most $k$, the source
instance has a constrained perfect matching.
\end{lemma}
\begin{proof}
Use the cover obtained above and select the unique primary set of each
group, supplied by Lemma~\ref{lem:cpm-selected-set}.  A selected set
containing $u$ supplies four primary paths for $u$, one per occurrence
row, and all end at primary element vertices of $u$ by
Lemma~\ref{lem:cpm-agreement}.  Conversely, all four primary element
arms of $u$ are used once and can only be used by these primary paths.
Thus the number of selected sets containing $u$ is exactly one.
The selected sets partition $U$ and contain exactly one set per group.
\end{proof}

Combining Lemmas~\ref{lem:cpm-forward} and~\ref{lem:cpm-reverse} gives
the reduction.  The source problem is \NP-complete and the construction
has polynomial size and pathwidth at most $287$
(Lemmas~\ref{lem:cpm-source-hard} and~\ref{lem:cpm-pathwidth}).
This proves Theorem \ref{thm:cover-tw-paraNPh}.

\section{Algorithm Parameterized by Treewidth plus Diameter}
\label{sec:algo-treewidth-plus-diam}

\newcommand{\sfP}{\mathsf{P}}
\newcommand{\sfQ}{\mathsf{Q}}
\newcommand{\sfR}{\mathsf{R}}
\newcommand{\sfS}{\mathsf{S}}
\newcommand{\calR}{\mathcal{R}}
\newcommand{\sfC}{\mathsf{C}}

\newcommand{\dd}{\mathsf{d}}

\renewcommand{\top}{+}
\renewcommand{\bot}{-}

\newcommand{\bfa}{\mathbf{a}}
\newcommand{\bfb}{\mathbf{b}}
\newcommand{\bfv}{\mathbf{v}}
\newcommand{\bfx}{\mathbf{x}}
\newcommand{\bfy}{\mathbf{y}}

\newcommand{\DC}{\textsf{DC}}
\renewcommand{\DP}{\textsf{DP}}

\newcommand{\opp}{{\mathsf{opp}}}
\newcommand{\prof}{{\mathsf{prof}}}
\newcommand{\ord}{{\mathsf{ord}}}

\newcommand{\proj}{\textsf{proj}}

In this section, we give a dynamic programming algorithm solving
\SoPC{} in \XP{} time parameterized by treewidth plus diameter.

As stated in Section~\ref{sec:intro}, the states of our algorithm describe the ``traces'' of paths on the bags of the decomposition, which consist of how paths meet a bag, which connections should be realized outside the bag, and some distance information of their endpoints. 
Together with these traces, we store multiplicities indicating how many paths from the solution yield a given trace.
Locally consistent traces that need not admit a global
completion are allowed; however, it is guaranteed that at the root of the decomposition, every trace is an isometric path. 

Let us also recall that, although this approach appears to follow standard techniques, it requires special effort to achieve the \XP{} time bound.
Most notably, instead of storing the exact endpoints of the traces, we must limit ourselves to the \emph{distance profiles} of these endpoints, which are vectors indicating their distances to the vertices of a given bag of the decomposition.
Handling these require particular attention for the compatibility of states, and is responsible of most of the complexity in proving the correctness of our algorithm.

\subsection{Objects of the dynamic programming}\label{sec:DP:objects}

We work with a finite, simple, connected, unweighted graph $G$ on $n$ vertices and diameter $\diam$.
Note that we may assume that $n\geq2$ and $\diam\geq2$, as otherwise the instance is trivial.
All distances $d(u,v)$ are distances in $G$, including when $u,v$ are considered within $G_t$ for some node~$t$.

Let $(T,(X_t)_{t\in V(T)})$ be a rooted nice tree decomposition of width $w$,
with empty root and leaf bags and introduce vertex, forget vertex, and
join nodes. The vertices of $G$ are labeled from $1$ to $n$.
For a node $t$, let $G_t$ be the graph induced by the vertices appearing
in the subtree rooted at $t$, and put $I_t=V(G_t)\setminus X_t$ and
$O_t=V(G)\setminus V(G_t)$ to be the \emph{interior} and \emph{outside} vertices w.r.t~$t$. The bag $X_t$ separates $I_t$ from $O_t$.

\paragraph{Signed distance profiles.}
Suppose $X_t\neq\emptyset$, and list its vertices as $x_1,\ldots,x_q$
in increasing label order. For $a\notin X_t$, its \emph{signed
distance profile} is
\[
 \prof_t(a)=(s,d_1,\ldots,d_q)
\]
where $d_i=d(a,x_i)$, and
where $s=-$ if $a\in I_t$ and $s=+$ if $a\in O_t$.
Every distance coordinate lies in $\{1,\ldots,\diam\}$, because $a\notin X_t$.
Consequently, there are at most $2\diam^q$ signed profiles at~$t$.

Write $\DP(t)$ for the set of profiles actually represented by vertices
of $G\setminus X_t$. Write $\DP^+(t)$ and $\DP^-(t)$ for its
positive and negative subsets. For a profile $\mathbf a$, write $\mathbf a[0]$ for its sign and
$\mathbf a[x_i]$ for its $i$th distance coordinate. Its \emph{distance class} is
\[
 V_t(\mathbf a)=\{a\notin X_t:\prof_t(a)=\mathbf a\}.
\]
In the following, we will only consider nonempty classes. 
For an empty bag, we set $\DP(t)=\emptyset$.

\begin{lemma}\label{lem:exact-distance}
Let $X_t\neq\emptyset$ and $\mathbf a\in\DP^+(t)$.
For every $a\in V_t(\mathbf a)$ and $b\in V(G_t)$,
\[
 d(a,b)=\min_{x\in X_t}\bigl(\mathbf a[x]+d(x,b)\bigr).
\]
In particular, the left-hand side is independent of the representative $a$.
\end{lemma}
\begin{proof}
Every $a$--$b$ path meets $X_t$. Taking a bag vertex on a shortest such
path, and applying the triangle inequality in the opposite direction,
gives $d(a,b)=\min_{x\in X_t}(d(a,x)+d(x,b))$.
Using $\mathbf a[x]=d(a,x)$ proves the equality.
\end{proof}

We introduce common notation for distances from endpoints represented either explicitly or by their profiles.

For $\mathbf a\in\DP^+(t)$ and $v\in V(G_t)$, define
\[
 \delta_t(\mathbf a,v)=\min_{x\in X_t}\bigl(\mathbf a[x]+d(x,v)\bigr).
\]
By Lemma~\ref{lem:exact-distance}, $\delta_t(\mathbf a,v)=d(a,v)$ for every representative $a\in V_t(\mathbf a)$. In particular, $\delta_t(\mathbf a,v)=\mathbf a[v]$ for $v\in X_t$.
For $\eta\in V(G_t)\cup\DP^+(t)$, set
\[
 H^t_\eta(v)=
 \begin{cases}
 d(\eta,v),&\eta\in V(G_t),\\
 \delta_t(\eta,v),&\eta\in\DP^+(t).
 \end{cases}
\]
For $\eta\in X_t\cup\DP(t)$, its distance function on the bag is
\[
 h^t_\eta(v)=
 \begin{cases}
 d(\eta,v),&\eta\in X_t,\\
 \eta[v],&\eta\in\DP(t).
 \end{cases}
\]
The second definition also applies to negative profiles, but only at bag vertices.

\paragraph{Traces of paths on bags and subtrees.}
Traces describe portions of a path and its endpoints, together with the
locations of the missing connections. See Figure~\ref{fig:DP:traces}.
Throughout, intersecting a path with a vertex set means retaining its
maximal contiguous subpaths on that set, with inherited order and
orientation.

\begin{figure}
 \centering
 \includegraphics[scale=1.2]{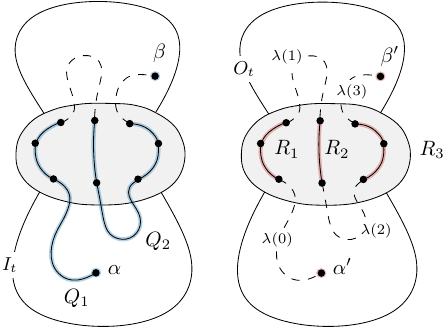}
 \caption{A subtree-trace (left) and its bag projection (right) which yields a bag-trace.
 Dashed connections through the exterior are not part of the subtree-trace.
 In the bag-trace, links record whether the connecting portion is above (its value is $\top$)
 or below (its value is $\bot$) the bag.}
 \label{fig:DP:traces}
\end{figure}

A \emph{subtree-trace at $t$} is a tuple $\sfQ=(\alpha,\calQ,\beta)$,
where $\alpha,\beta\in V(G_t)\cup\DP^+(t)$ and
$\calQ=(Q_1,\ldots,Q_k)$ is a nonempty ordered family of oriented,
nonempty, pairwise vertex-disjoint paths of $G_t$.
(We stress the fact that these vertex-disjoint paths are meant to represent the trace of a single isometric path on $G_t$.)
In addition, we require a number of conditions that we detail now; in the remaining, write $u_i$ and $v_i$ for the first and last vertices of $Q_i$.
The first set of conditions, that we call \emph{distance conditions}, consists of the following:
\begin{itemize}
 \item Along every oriented edge of every path $Q_i$, $H^t_\alpha$ increases
 by one and $H^t_\beta$ decreases by one.
 \item For $1\leq i<k$,
 \[
 H^t_\alpha(v_i)<H^t_\alpha(u_{i+1}),\qquad
 H^t_\beta(v_i)>H^t_\beta(u_{i+1}).
 \]
\end{itemize}
Then we require the following \emph{boundary conditions}:
\begin{itemize}
 \item For $1\leq i<k$, both $v_i$ and $u_{i+1}$ lie in $X_t$.
 \item If $\alpha\in V(G_t)$, then $\alpha=u_1$; if $\beta\in V(G_t)$, then $\beta=v_k$.
 \item If $\alpha\in\DP^+(t)$, then $u_1\in X_t$; if $\beta\in\DP^+(t)$, then $v_k\in X_t$.
\end{itemize}
In particular, if $k\geq2$, every $Q_i$ meets $X_t$.
The conditions on $u_1$ and $v_k$ for positive endpoints are essential:
any portion continuing outside $G_t$ must end at a bag vertex.

A subtree-trace is \emph{partial} if it meets $X_t$ and \emph{complete}
otherwise. 
By the conditions above, a complete trace has $k=1$ and
explicit endpoints $\alpha=u_1$, $\beta=v_1$. The distance conditions
then imply that $Q_1$ represents a possible isometric $\alpha$--$\beta$ path.
These two cases exhaust all subtree-traces, including at an empty bag,
where all traces are complete.

A trace needs not be the restriction of an actual isometric path:
the missing connections are not required to exist. The dynamic program
uses exactly the local conditions above. Conversely, the nonempty
restriction of any oriented isometric path to $G_t$ is a subtree-trace,
with any endpoint outside $G_t$ replaced by its positive profile.
Indeed, distances from its endpoints are monotone along the path,
and every crossing of the boundary uses $X_t$.

A \emph{bag-trace at $t$} is a tuple
$\sfR=(\alpha,\calR,\beta,\lambda)$, where
$\alpha,\beta\in X_t\cup\DP(t)$,
$\calR=(R_1,\ldots,R_k)$ is a nonempty ordered family of oriented,
nonempty, pairwise vertex-disjoint paths of $G[X_t]$, and $\lambda$ is a function.
In addition, we require a number of conditions that we detail now; in the remaining, write $u_i$ and $v_i$ for the first and last vertices of $R_i$.
The first set of conditions, we call \emph{distance conditions}, consist of the following:
\begin{itemize}
 \item Along each oriented edge, $h^t_\alpha$ increases by one and
 $h^t_\beta$ decreases by one.
 \item For $1\leq i<k$,
 \[
 h^t_\alpha(v_i)<h^t_\alpha(u_{i+1}),\qquad
 h^t_\beta(v_i)>h^t_\beta(u_{i+1}).
 \]
\end{itemize}
Then we require the following \emph{boundary condition}: 
\begin{itemize}
 \item if $\alpha\in X_t$, then $\alpha=u_1$; if $\beta\in X_t$, then $\beta=v_k$.
\end{itemize}
Finally, we require some \emph{linking conditions} that we detail now.
The \emph{genuine link indices} are the indices $1,\ldots,k-1$ to which we add $0$ if $\alpha\in\DP(t)$, and $k$ if $\beta\in\DP(t)$.
We note $L(\sfR)$ the set of these indices.
We require that $\lambda$ totally maps $L(\sfR)$ to $\{\top,\bot\}$, indicating whether the consecutive subpaths are aimed to be connected through $G-G_t$ (top) or through $G_t-X_t$ (bottom).
Here, the index $i\in\{1,\ldots,k-1\}$ is considered for the connection from $v_i$ to $u_{i+1}$; index $0$ connects the first endpoint to $u_1$; and index $k$ connects
$v_k$ to the second endpoint. 
There is no endpoint link when that
endpoint is itself a bag vertex.
The function $\lambda$ is called \emph{linking function}, and by \emph{links} we designate the pairs $v_i,u_{i+1}$.
We say that such a link is a \emph{top link} if $\lambda(i)=\top$, and \emph{bottom link} if $\lambda(i)=\bot$. 
We say that a link is \emph{incident with a bag vertex $x$} if one
of its attachment vertices $u_1,v_k,v_i,u_{i+1}$ is $x$.
Finally, we require that the link incident to a positive profile is a top link, and that the one incident to a negative profile is a bottom link.

In the remaining of the section, when considering a bag- or subtree-trace, by \emph{bag vertex} (resp.~\emph{bag edge}) we mean a vertex (resp.~edge) of its its ordered family of subpaths, and by  \emph{gap} we informally refer to the separations between consecutive subpaths of the family. 

\paragraph{Partial solutions and states.}
A \emph{partial solution at $t$} is a pair $\sfP=(\calP,\mu)$,
where $\calP$ is a finite set of subtree-traces and
$\mu\colon \calP\to\mathbb Z_{\geq1}$ is a multiplicity function.
The subpaths of these traces must cover $V(G_t)$, and $\mu$ must map complete subtree-traces to $1$. 
Its \emph{order} is $\ord(\sfP)=\sum_{\sfQ\in\calP}\mu(\sfQ)$.
Equivalently, $\sfP$ can be regarded as a finite multiset of trace occurrences.

A \emph{state at $t$} is a pair $\sfS=(\calS,\nu)$, where $\calS$ is
a set of bag-traces and $\nu\colon\calS\to\mathbb Z_{\geq1}$ is a multiplicity function.
Their subpaths must cover $X_t$. Its \emph{order} is 
\[
 \ord(\sfS)=\sum_{\sfR\in\calS}\nu(\sfR).
\]
The empty state is $\sfS_0=(\emptyset,\emptyset)$, where the second
empty set denotes the empty function. It is the unique state of an empty bag.
Multiplicities are necessary because distinct path occurrences can have
the same bag-trace and can later develop different traces; see Figure~\ref{fig:DP:restriction} (right).
A set alone
would lose the number of occurrences available for these extensions.

\paragraph{Projection, realizability, and values.}
The projection $\proj_t(\sfQ)$ of a subtree-trace
$\sfQ=(\alpha,\calQ,\beta)$ on (the bag of) $t$ is empty if it is complete.
Otherwise, form $\calR$ by intersecting the subpaths of $\calQ$ with $X_t$,
in their order and orientation. Replace each explicit endpoint in $I_t$
by its negative profile, retain bag endpoints and positive profiles,
and define the links as follows. Consecutive bag portions have a bottom
link if they are joined within the same $Q_i$, and a top link if they
belong to different $Q_i$. Endpoint links are bottom for explicit
endpoints in $I_t$ and top for positive endpoints.
This gives $\proj_t(\sfQ)=(\alpha',\calR,\beta',\lambda)$.

This tuple is a bag-trace. 
For an explicit interior endpoint, conversion
to its profile preserves its exact distance to every bag
vertex. Thus all endpoint distance functions are unchanged on the bag.
Along a subpath the required unit increments remain valid; across a
removed interior portion or a gap, strict inequalities follow from the
subtree-trace conditions. The endpoint and link conditions follow from
the boundary conditions and the construction.

The projection of a partial solution is obtained by projecting its subtree-traces occurrence by occurrence, discarding empty
projections, and aggregating multiplicities of equal bag-traces. Thus, the multiplicity function of the resulting state is
\[
\nu(\sfR)=\sum_{\substack{\sfQ\in\calP\\\proj_t(\sfQ)=\sfR}}\mu(\sfQ).
\]
Coverage of $G_t$ implies coverage of $X_t$, so this is a state.
We say that $\sfP$ \emph{realizes} $\sfS$ when its projection is $\sfS$,
and define
\[
 \dd_t[\sfS]=\min\{\ord(\sfP):\sfP\text{ realizes }\sfS\},
\]
with value $+\infty$ if no realization exists, or its order exceed $n$.

\subsection{Computation of the values}

We show how the values of states can be computed in bottom-up fashion.

\subsubsection{Leaf Node}

Let $t$ be a leaf node.
Then $X_t=\emptyset$.
Recall that bag-traces are required to contain at least one vertex, by definition.
Thus the only possible state for $t$ consists of an empty set of bag-traces, and an empty multiplicity function.
Its only realizing partial solution is the empty multiset of subtree-traces.
The optimal value of this state is 0 since $G_t$ is the empty graph.

We derive the next lemma.

\begin{lemma}\label{lem:dp-leaf-node}
    If $t$ is a leaf node then $\dd_t[\sfS]=0$ for $\sfS=(\emptyset, \emptyset)$, and $\dd_t[\sfS]=+\infty$ otherwise.
\end{lemma}

\subsubsection{Introduce Vertex Node}

Let $t$ have child $t'$ with $X_t=X_{t'}\cup\{x\}$ and
$V(G_t)=V(G_{t'})\cup\{x\}$. In particular,
$I_t=I_{t'}$ and $x$ has no neighbor in $I_{t'}$.

Before explaining how the value of an introduce vertex node is computed, we introduce a series of operations related to the deletion of vertex of the bag.
While conducted in the natural way, these definition are lengthy due to the complex nature of these objects.
These operations are better understood with Figures~\ref{fig:DP:restriction} and \ref{fig:DP:introduce}.

\paragraph{Restriction of profiles.}
If $X_{t'}\neq\emptyset$, define
\[
 \rho_{t\to t'}(\mathbf a)[0]=\mathbf a[0],\qquad
 \rho_{t\to t'}(\mathbf a)[v]
 =\mathbf a[v]
 \quad \text{for every }v\in X_{t'}.
\]
For $a\notin X_t$, this gives
$\rho_{t\to t'}(\prof_t(a))=\prof_{t'}(a)$:
the sign and all retained distance coordinates are unchanged.
Restriction only deletes the coordinate corresponding to $x$.
The vertex $x$, when converted to a child profile, has positive sign.

\paragraph{Introduce-admissibility}
A bag-trace at $t$ is \emph{introduce-admissible} if no bottom link is
incident with $x$. A state is introduce-admissible if all its traces are.
Every realizable state has this property, because a bottom connection
at $x$ would require an edge from $x$ to $I_t$.

\paragraph{Restriction of bag-traces.}
For an introduce-admissible bag-trace
$\sfR=(\alpha,\calR,\beta,\lambda)$, its restriction
$\operatorname{res}_x(\sfR)$ is defined by deleting $x$ from its ordered,
oriented subpaths and retaining the maximal nonempty contiguous pieces.
If none remain, the restriction is empty. Otherwise $X_{t'}\neq\emptyset$,
and the first endpoint becomes
\[
 \alpha'=
 \begin{cases}
 \alpha,&\alpha\in X_{t'},\\
 \prof_{t'}(x),&\alpha=x,\\
 \rho_{t\to t'}(\alpha),&\alpha\in\DP(t).
 \end{cases}
\]
Define $\beta'$ symmetrically. Connections unaffected by deleting $x$ are not changed. Those going through $x$ are top:
this includes the new link when a subpath splits, the merged link when
a singleton $x$ disappears, and the endpoint link when an endpoint $x$
is replaced by its profile. 

Note that every nonempty restriction is a bag-trace. Indeed, exact distances from retained
endpoint profiles are unchanged on the remaining
bag, and the same is true for distances from an endpoint $x$ replaced by its profile.
All edge and gap conditions therefore persist, and the endpoint and
link rules are consistent with the new signs.
If $X_{t'}=\emptyset$, every bag-trace restricts to the empty tuple and
no profile conversion is performed.

\begin{figure}
 \centering
 \includegraphics[scale=1.2]{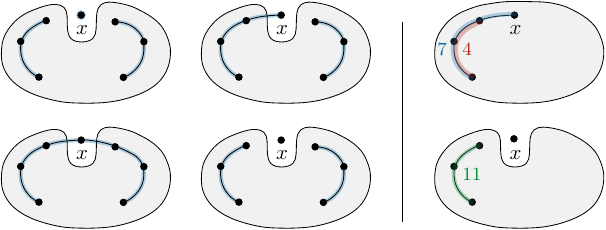}
 \caption{Deleting a bag vertex may remove a singleton subpath (top left),
 shorten a subpath (top middle), split a subpath (bottom left), or may have no incidence on subpaths (bottom middle).
 In any cases, the deletion of a vertex has implications on the distance profiles. 
 Moreover, the restriction of two distinct bag-traces can give a same trace, and the count of the multiplicities is conserved in the restriction of a state (top right and bottom right).
 At an introduce restriction,
 every connection newly passing through the deleted vertex is top.}
 \label{fig:DP:restriction}
\end{figure}

For an admissible state $\sfS=(\calS,\nu)$, define
$\operatorname{res}_x(\sfS)=(\calS',\nu')$ by discarding empty
restrictions and setting
\[
 \nu'(\sfR')=
 \sum_{\substack{\sfR\in\calS\\\operatorname{res}_x(\sfR)=\sfR'}}\nu(\sfR).
\]
The surviving traces cover $X_{t'}$, so this is a state.
Set
\[
 E_x(\sfS)=\{\sfR\in\calS:\operatorname{res}_x(\sfR)=\emptyset\},
 \qquad
 \varepsilon_x(\sfS)=\sum_{\sfR\in E_x(\sfS)}\nu(\sfR).
\]
Thus $E_x(\sfS)$ consists precisely of traces whose sole bag subpath
is the singleton $x$, and $\varepsilon_x(\sfS)$ record how much they participate in the order of the state.

\paragraph{Restriction of subtree-traces.}
The restriction of subtree-traces is defined similarly as for the restriction of bag-traces, as follows.
First delete $x$ from the subpaths, retaining the
ordered maximal nonempty pieces. If no piece remains, discard the trace.
Otherwise retain explicit endpoints in $G_{t'}$, replace an endpoint
$x$ by $\prof_{t'}(x)$, and restrict positive profiles with $\rho_{t\to t'}$.
When $X_{t'}=\emptyset$, a trace containing $x$ consists only of that
singleton: $x$ has no neighbor in $I_t$, and the boundary conditions
prevent a second component or an interior endpoint attached to $x$.
All other traces are complete and remain unchanged.

For a partial solution $\sfP=(\calP,\mu)$, define
$\operatorname{res}_x(\sfP)=(\calP',\mu')$ by restricting subtree-traces occurrence by occurrence and discarding empty
restrictions, and setting
\[
 \mu'(\sfQ')=
 \sum_{\substack{\sfQ\in\calP\\\operatorname{res}_x(\sfQ)=\sfQ'}}\nu(\sfQ).
\]

It can be readily checked that the following properties hold. 

\begin{lemma}\label{lem:restriction-of-bag-trace}
The restriction of an introduce-admissible state at $t$ is a state at $t'$.
\end{lemma}

\begin{lemma}\label{lem:restriction-of-subtree-trace}
The restriction of a partial solution at $t$ is a partial solution at $t'$.
\end{lemma}

\begin{lemma}\label{lem:restriction-preserves-realizability}
If $\sfP$ realizes $\sfS$ at $t$, then its restriction $\operatorname{res}_x(\sfP)$ realizes
$\operatorname{res}_x(\sfS)$ at $t'$.
\end{lemma}

\paragraph{Lifting a restricted trace.}

The following lemma supplies the converse to restriction which we will need in the proof of our recurrence. 

\begin{lemma}\label{lem:introduce-lifting}
Let $\sfR$ be an introduce-admissible bag-trace at $t$.
\begin{enumerate}
    \item If $\operatorname{res}_x(\sfR)\neq\emptyset$,
    then every subtree-trace $\sfQ'$ at $t'$ projecting
    to $\operatorname{res}_x(\sfR)$ can be lifted to
    a subtree-trace $\sfQ$ at $t$ projecting to $\sfR$.
    The lift retains every vertex of $\sfQ'$ and
    restricts back to $\sfQ'$.

    \item If $\operatorname{res}_x(\sfR)=\emptyset$,
    then $\sfR$ is the projection of a subtree-trace
    whose only subpath is the singleton $\{x\}$.
\end{enumerate}
\end{lemma}

\begin{proof}
For the first statement, retain all portions of $\sfQ'$
inside $G_{t'}$ and insert $x$, when prescribed by $\sfR$,
with the bag edges and position specified by $\sfR$.
Connections through the child interior are retained.

Explicit endpoints in $G_{t'}$ are retained.
An endpoint prescribed as $x$ by $\sfR$ becomes the
explicit vertex $x$.
An endpoint prescribed as a positive profile by $\sfR$
is replaced by that profile.

We first verify the consistency of negative endpoint
profiles.
Suppose an explicit endpoint $a$ of $\sfQ'$ lies in
$G_{t'}-X_{t'}$.
Since $X_{t'}$ separates $x$ from $I_{t'}$, we have
\[
 d(a,x)=\min_{z\in X_{t'}}\bigl(\prof_{t'}(a)[z]+d(z,x)\bigr).
\]
Thus the profile of $a$ w.r.t.~$t$ is determined by its profile w.r.t.~$t'$.
Since $I_t=I_{t'}$, it follows that
$V_t(\prof_t(a))=V_{t'}(\prof_{t'}(a))$.
Thus retaining $a$ gives the required parent profile.

We next verify the distance conditions.
On the retained vertices of $G_{t'}$, the endpoint
exact distance functions are unchanged.
This is immediate for retained explicit endpoints.
For a positive parent profile, choose a representative
outside $G_t$ and apply
Lemma~\ref{lem:exact-distance} at both nodes.
For an endpoint changed from $\prof_{t'}(x)$ to $x$,
the same lemma gives
\[
    d(x,v)
    =
    \delta_{t'}\bigl(\prof_{t'}(x),v\bigr)
\]
for every $v\in V(G_{t'})$.
Therefore all distance conditions within retained
portions are preserved.

Every newly inserted edge has both endpoints in $X_t$:
$x$ has no neighbor in the child interior.
The required increase by one along such an edge follows
from the distance conditions of $\sfR$, using the
endpoint-profile consistency established above.
Likewise, inequalities across newly created gaps follow
from the ordered bag vertices in $\sfR$.
For gaps inherited from $\sfQ'$, they follow from the
child distance conditions.
The same argument applies to both endpoints.

The construction also preserves the boundary conditions.
Component ends that still require an external continuation
belong to $X_t$, and explicit endpoints occupy the
prescribed ends.
The resulting components are vertex-disjoint paths:
the old components are vertex-disjoint, $x$ is new,
and it is inserted only once, in its prescribed position.

Thus the resulting tuple $\sfQ$ is a subtree-trace.
Its projection is $\sfR$, and restricting it to $t'$
recovers $\sfQ'$.

For the second statement, the only bag subpath of
$\sfR$ is $\{x\}$.
An explicit endpoint must therefore equal $x$.
A profile endpoint cannot be negative, since its
incident link would then be a bottom link at $x$.
Consequently, both endpoints are either $x$ or positive
profiles.
Using these endpoints and the singleton subpath $\{x\}$
defines the required subtree-trace.

If $X_{t'}=\emptyset$, every bag-trace at $t$ has
empty restriction, so only the second statement applies;
no profile at an empty bag is needed.
\end{proof}

\paragraph{Recurrence.}

We are now ready to prove the following recurrence.

\begin{figure}
 \centering
 \includegraphics[scale=1.2]{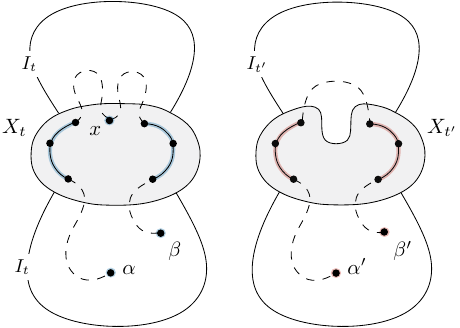}
 \caption{A bag-trace (left) and its restriction (right) in the context of an introduce-node.}
 \label{fig:DP:introduce}
\end{figure}

\begin{lemma}\label{lem:dp-introduce-node}
Let $t$ be an introduce vertex node with child $t'$,
and let $\sfS$ be a state at $t$.
If $\sfS$ is not introduce-admissible, then
\[
    \dd_t[\sfS]=+\infty.
\]
Otherwise, letting $\sfS'=\operatorname{res}_x(\sfS)$,
we have
\[
    \dd_t[\sfS]
    =
    \dd_{t'}[\sfS']
    +
    \varepsilon_x(\sfS).
\]
\end{lemma}

\begin{proof}
If $\sfS$ is not introduce-admissible, one of its
bag-traces requires a bottom connection incident with
$x$.
Since $x$ has no neighbor in $G_t-X_t$, no
subtree-trace can realize that bag-trace.
Thus $\dd_t[\sfS]=+\infty$.

Suppose now that $\sfS$ is introduce-admissible.
This situation is depicted in Figure~\ref{fig:DP:introduce}.

We first prove
\[
    \dd_t[\sfS]
    \geq
    \dd_{t'}[\sfS']+\varepsilon_x(\sfS).
\]
If $\dd_t[\sfS]=+\infty$, the inequality is immediate.
Otherwise, let $\sfP$ be a minimum-order partial solution realizing $\sfS$, and $\sfP'$ be its restriction to $t'$.
By Lemma~\ref{lem:restriction-preserves-realizability}, $\sfP'$ realizes $\sfS'$.

Recall that in the restriction to $t'$, an occurrence disappears precisely when its only subpath is the singleton $\{x\}$.
These are exactly the occurrences whose bag projections belong to $E_x(\sfS)$.
Consequently,
\[
    \ord(\sfP')
    =
    \ord(\sfP)-\varepsilon_x(\sfS).
\]
and so $\dd_{t'}[\sfS'] \leq \dd_t[\sfS]-\varepsilon_x(\sfS)$ leading to the desired inequality.

For the converse, suppose that $\dd_{t'}[\sfS']$
is finite and choose a minimum-order partial solution
$\sfP'$ realizing $\sfS'$.

We construct a partial solution for $t$ as follows.
First, we retain every complete occurrence of $\sfP'$ unchanged.
Then, for each bag-trace $\sfR'$ of $\sfS'$, we collect all
occurrences of $\sfP'$ projecting to $\sfR'$.
Their number is $\nu'(\sfR')$, and the definition of
state restriction gives
\[
    \nu'(\sfR')
    =
    \sum_{\substack{\sfR\in\calS\\
        \operatorname{res}_x(\sfR)=\sfR'}}
        \nu(\sfR).
\]
We may therefore partition these occurrences among the
parent bag-traces $\sfR$ restricting to $\sfR'$, allocating
exactly $\nu(\sfR)$ occurrences to each such $\sfR$.
We lift each allocated occurrence to its assigned parent
bag-trace using Lemma~\ref{lem:introduce-lifting}.
This replaces each allocated occurrence by exactly one
occurrence and preserves all its covered vertices.
For every $\sfR\in E_x(\sfS)$, additionally insert
$\nu(\sfR)$ occurrences of the singleton subtree-trace
provided by the second statement of Lemma~\ref{lem:introduce-lifting}.
This completes the construction of $\sfP$.

Note that every vertex of $G_{t'}$ remains covered.
Moreover, $\sfS$ covers $X_t$, so some prescribed parent bag-trace contains $x$.
Its corresponding lifted or newly inserted occurrence covers $x$.
Thus $\sfP$ covers $G_t$.

By the allocation of occurrences, exactly $\nu(\sfR)$
occurrences project to each $\sfR\in\calS$.
Hence $\sfP$ realizes $\sfS$.
Its order is
\[
    \ord(\sfP)
    =
    \ord(\sfP')
    +
    \sum_{\sfR\in E_x(\sfS)}\nu(\sfR)
    =
    \dd_{t'}[\sfS']+\varepsilon_x(\sfS).
\]
Therefore,
\[
    \dd_t[\sfS]
    \leq
    \dd_{t'}[\sfS']+\varepsilon_x(\sfS).
\]
This concludes the proof.
\end{proof}

\subsubsection{Forget Vertex Node}

Let $t$ be a forget vertex node with child $t'$.
Then $X_t=X_{t'}\setminus\{x\}$ and $G_t=G_{t'}$ for some vertex $x\in X_{t'}$.
Moreover, $x$ has no neighbor in $G-G_t$.

Unlike an introduce node, a forget node does not remove
a vertex from the subtree graph.
It only moves $x$ from the bag to the interior of the
subtree.
Accordingly, subtree-traces retain their subpaths, whereas
their bag projections must be updated.

\paragraph{Restriction of profiles.}

Assume that $X_t\neq\emptyset$.
For a {signed distance profile} $\mathbf{a}$ at $t'$,
define $\rho_{t'\to t}(\mathbf{a})$ by
\[
    \rho_{t'\to t}(\mathbf{a})[0]=\mathbf{a}[0],
    \qquad
    \rho_{t'\to t}(\mathbf{a})[v]
    =
\mathbf{a}[v]
    \quad\text{for every }v\in X_t.
\]
Thus, we delete the coordinate corresponding to $x$ and
leave all remaining distance coordinates unchanged.
Note that for every vertex $a\notin X_{t'}$, we have
\[
    \rho_{t'\to t}\bigl(\prof_{t'}(a)\bigr)
    =
    \prof_t(a).
\]

\paragraph{Forget-admissibility.}

Let $\sfR'$ be a bag-trace at $t'$.
We call $\sfR'$ \emph{forget-admissible} if no top link
is incident with $x$.
This condition expresses that, after forgetting $x$,
a subpath cannot leave the subtree graph through $x$.

\paragraph{Restriction of bag-traces.}

For a forget-admissible bag-trace
\[
    \sfR'=(\alpha',\calR',\beta',\lambda'),
\]
we define its \emph{forget-restriction},
denoted by $\operatorname{forg}_x(\sfR')$, as follows.

First, delete $x$ from the ordered and oriented subpaths
of $\calR'$, possibly shortening or splitting a subpath of $\calR'$, and keeping the inherited order and orientation.
Let $\calR$ be the obtained set.
If $\calR=\emptyset$, set
\[
    \operatorname{forg}_x(\sfR')=\emptyset.
\]
Otherwise, $X_t\neq\emptyset$, and the first endpoint
of the restricted trace is
\[
    \alpha=
    \begin{cases}
        \alpha' & \text{if}\ \alpha'\in X_t;\\
        \prof_t(x) & \text{if}\ \alpha'=x;\\
        \rho_{t'\to t}(\alpha') & \text{if}\ \alpha'\in\DP(t').
    \end{cases}
\]
In particular, $\prof_t(x)$ is a negative profile.
The second endpoint $\beta$ is defined symmetrically.

The linking function $\lambda$ is obtained by retaining the inherited links if they did not involve $x$, and putting a bottom link otherwise.
More explicitly:
\begin{itemize}
    \item If deleting $x$ splits one subpath into
    two nonempty subpaths, the link between them is bottom;

    \item If deleting the singleton subpath $\{x\}$
    joins two consecutive links into one, the resulting link is bottom.
    Both original links are bottom by forget-admissibility;

    \item If deleting $x$ shortens a subpath at an end incident with a link, that link remains bottom, with its attachment moved to the new end;

    \item If an endpoint equal to $x$ is replaced by
    $\prof_t(x)$, its connection to the first or last
    surviving subpath, respectively, is bottom.
\end{itemize}
These rules apply symmetrically at the two ends of the
ordered trace.
We then put
\[
    \operatorname{forg}_x(\sfR')
    =
    (\alpha,\calR,\beta,\lambda).
\]

All retained distance coordinates are unchanged, and an endpoint $x$ converted to its profile still has distance $\prof_t(x)[v]=d(x,v)$ to each $v\in X_t$. Thus the distance conditions are preserved, and every nonempty forget-restriction is a bag-trace at $t$.
If $X_t=\emptyset$, every bag-trace at $t'$ consists of the singleton path $\{x\}$, and its forget-restriction is therefore empty.

\paragraph{Compatibility of states.}

A state $\sfS'=(\calS',\nu')$ at $t'$ is
\emph{forget-admissible} if all its bag-traces are
forget-admissible.
Its \emph{forget-restriction} is the state
\[
    \operatorname{forg}_x(\sfS')=(\calS,\nu),
\]
where
\[
    \calS
    =
    \bigl\{
        \operatorname{forg}_x(\sfR')
        :
        \sfR'\in\calS',
        \ \operatorname{forg}_x(\sfR')\neq\emptyset
    \bigr\},
\]
and, for every $\sfR\in\calS$,
\[
    \nu(\sfR)
    =
    \sum_{\substack{\sfR'\in\calS'\\
        \operatorname{forg}_x(\sfR')=\sfR}}
        \nu'(\sfR').
\]
Every vertex of $X_t$ remains covered, so this indeed
defines a state at $t$.

For a state $\sfS$ at $t$, let
\[
    F(\sfS)
    =
    \{
        \sfS' :
        \sfS'\text{ is a forget-admissible state at }t'\ \text{and}\ \operatorname{forg}_x(\sfS')=\sfS
    \}.
\]

\paragraph{Transport of partial solutions.}

We establish the relationship between realizations
on the two sides of a forget node.
We regard partial solutions
as multisets of subtree-trace occurrences.

\begin{lemma}\label{lem:forget-transport}
Let $t$ be a forget vertex node with child $t'$.
The following statements hold.
\begin{enumerate}
    \item If $\sfS'$ is a forget-admissible state at $t'$
    and $\sfP'$ realizes $\sfS'$, then there exists
    a partial solution $\sfP$ at $t$ realizing
    $\operatorname{forg}_x(\sfS')$ and having the
    same order as $\sfP'$.

    \item If $\sfP$ realizes a state $\sfS$ at $t$,
    then there exist a state $\sfS'\in F(\sfS)$
    and a partial solution $\sfP'$ at $t'$ realizing
    $\sfS'$ and having the same order as $\sfP$.
\end{enumerate}
\end{lemma}

\begin{proof}
We start with the first statement. 
Consider the tuples obtained by keeping the subpaths and explicit endpoints of every subtree-trace of $\sfP'$, and, if $X_t\neq\emptyset$, replacing each positive endpoint profile $\mathbf{a}'$ by $\rho_{t'\to t}(\mathbf{a}')$.

We verify that the resulting tuples are subtree-traces at $t$.
If $X_t\neq\emptyset$, choose a representative $a\in V_{t'}(\mathbf a')$ for each positive endpoint profile. By Lemma~\ref{lem:exact-distance}, for every $v\in V(G_t)$,
\[
 \delta_t\bigl(\rho_{t'\to t}(\mathbf a'),v\bigr)
 =d(a,v)=\delta_{t'}(\mathbf a',v).
\]
Distances from explicit endpoints are also unchanged. Hence all distance conditions are preserved.

As of the boundary conditions, two situations arise depending on whether a trace still intersects the bag, or no longer does so. 
In the first scenario, only the first and third boundary conditions could be violated, when $x$ is the endpoint of a subpath in a bag-trace, in which case it must be incident to a top link. This situation, however, is excluded by the forget-admissibility of bag-traces in $\sfP'$.
In the other scenario, the bag-trace may already be avoided by $X_{t'}$, and hence remains complete.
Otherwise, its intersection with $X_{t'}$ was exactly $\{x\}$.
The boundary conditions and forget-admissibility imply that it consists of a single subpath with both endpoints explicit in $G_t$.
It therefore becomes complete.
In particular, when $X_t=\emptyset$, all resulting
occurrences are complete and no positive endpoint
profile needs to be converted.

So far, we have retained every occurrence and every covered vertex.
After aggregating equal resulting traces, summing the multiplicities, we obtain a partial solution $\sfP$ of the same order as $\sfP'$.
By construction, its bag projection is
$\operatorname{forg}_x(\sfS')$, as desired.

\medskip

We now turn to the second statement.
First retain every subpath and
explicit endpoint of every bag-trace of $\sfP$.
For each positive endpoint profile $\mathbf{a}$ at $t$,
choose a representative $a\in V_t(\mathbf{a})$ and
replace the profile by $\prof_{t'}(a)$.
Such a representative remains outside $G_{t'}$ because
$G_t=G_{t'}$.
If $X_t=\emptyset$, all occurrences are complete,
so no such replacement is necessary.

For a replaced endpoint, its exact distance function on $G_t$ is unchanged, by Lemma~\ref{lem:exact-distance}.
Hence all distance conditions are preserved.
The boundary conditions remain valid because
$X_t\subseteq X_{t'}$.
We obtain a partial solution $\sfP'$ at $t'$ with
the same order and coverage as $\sfP$.

Let $\sfS'$ be its projection on $X_{t'}$.
No top link of this projection can be incident with $x$:
such a link would correspond to a component end that
must continue outside $G_t$, whereas $x\notin X_t$
and the original subtree-trace satisfies the boundary
conditions at $t$.
Thus $\sfS'$ is forget-admissible.

Finally, deleting $x$ from these bag projections and
restricting their endpoint profiles recovers the
original projections on $X_t$.
After aggregating multiplicities, this gives
\[
    \operatorname{forg}_x(\sfS')=\sfS.
\]
Therefore $\sfS'\in F(\sfS)$, as required.
\end{proof}

\paragraph{Recurrence.}

We are now ready to prove our recurrence.

\begin{lemma}\label{lem:dp-forget-node}
If $t$ is a forget vertex node with child $t'$
and $\sfS$ is a state at $t$, then
\[
    \dd_t[\sfS]
    =
    \min_{\sfS'\in F(\sfS)}
        \dd_{t'}[\sfS'].
\]
\end{lemma}

\begin{proof}
Let $\sfS'\in F(\sfS)$.
If $\dd_{t'}[\sfS']$ is finite, choose a minimum-order
partial solution realizing $\sfS'$.
By the first statement of
Lemma~\ref{lem:forget-transport}, it yields a partial
solution realizing $\sfS$ with the same order.
Hence
\[
    \dd_t[\sfS]\leq\dd_{t'}[\sfS'].
\]
The inequality is immediate if the child value is
infinite.
Taking the minimum over $F(\sfS)$ gives
\[
    \dd_t[\sfS]
    \leq
    \min_{\sfS'\in F(\sfS)}\dd_{t'}[\sfS'].
\]

Conversely, suppose that $\dd_t[\sfS]$ is finite,
and choose a minimum-order partial solution $\sfP$
realizing $\sfS$.
By the second statement of
Lemma~\ref{lem:forget-transport}, there exist
$\sfS'\in F(\sfS)$ and a partial solution $\sfP'$
realizing $\sfS'$ with
\[
    \ord(\sfP')
    =
    \ord(\sfP)
    =
    \dd_t[\sfS].
\]
Therefore,
\[
    \min_{\sfS'\in F(\sfS)}\dd_{t'}[\sfS']
    \leq
    \dd_t[\sfS].
\]
This inequality is immediate when $\dd_t[\sfS]$ is infinite.
Together, the two inequalities prove the desired equality.
\end{proof}

\subsubsection{Join Node}

Let $t$ be a join node with children $t_1$ and $t_2$;
then $X_t=X_{t_1}=X_{t_2}$.

\paragraph{Compatibility of traces.}

We first define compatibility at the level of individual traces.
Let
\[
    \sfR_1=(\alpha_1,\calR_1,\beta_1,\lambda_1),
    \qquad
    \sfR_2=(\alpha_2,\calR_2,\beta_2,\lambda_2),
    \qquad
    \sfR=(\alpha,\calR,\beta,\lambda)
\]
be bag-traces at $t_1$, $t_2$, and $t$, respectively.
We say that the triple $(\sfR_1,\sfR_2,\sfR)$ is
\emph{join-compatible} if the following conditions hold:
\begin{itemize}
    \item The three traces have the same subpaths, with a common
    ordering and orientation satisfying their distance conditions.
    In particular, $\calR_1=\calR_2=\calR$.

    \item The endpoints $\alpha_1,\alpha_2,\alpha$ are either
    all equal to the same vertex of $X_t$, or are all distance
    profiles with identical distance coordinates and signs
    satisfying
    \[
        (\alpha_1[0],\alpha_2[0],\alpha[0])
        \in
        \bigl\{(+,+,+),(-,+,-),(+,-,-)\bigr\}.
    \]
    The same condition holds for $\beta_1,\beta_2,\beta$.

    \item For every genuine link index $i$ in the common domain
    $L(\sfR_1)=L(\sfR_2)=L(\sfR)$,
    \[
        \bigl(\lambda_1(i),\lambda_2(i),\lambda(i)\bigr)
        \in
        \bigl\{
            (\top,\top,\top),
            (\bot,\top,\bot),
            (\top,\bot,\bot)
        \bigr\}.
    \]
\end{itemize}
Thus, an endpoint outside the bag may lie in the interior of
at most one child subtree; it is interior to the parent subtree
exactly when it is interior to one of the children.
Similarly, a link may be a bottom link in at most one child,
and it is a bottom link in the parent exactly when it is a
bottom link in one of the children.
See Figure~\ref{fig:join-compatibility} for an illustration.

\begin{figure}
    \centering
    \includegraphics[scale=1.2]{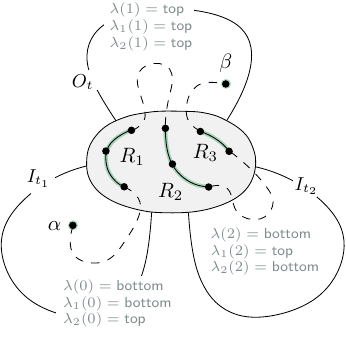}
    \caption{A join-compatible triple of bag-traces.
    The endpoint represented by $\alpha$ has a negative profile
    at $t_1$ and $t$, and a positive profile at $t_2$.
    Its link to $R_1$ is a bottom link at $t_1$ and $t$,
    and a top link at $t_2$.
    The three traces agree on the profile of the other endpoint
    and on the link between $R_1$ and $R_2$.}
    \label{fig:join-compatibility}
\end{figure}

\paragraph{Compatibility of states.}

We now extend this definition to states, taking multiplicities
into account.
Let $\sfS=(\calS,\nu)$ be a state for $t$, and let
$\sfS_1=(\calS_1,\nu_1)$ and $\sfS_2=(\calS_2,\nu_2)$
be states for $t_1$ and $t_2$, respectively.
We say that $\sfS$ is \emph{join-compatible} with
$(\sfS_1,\sfS_2)$ if there exist nonnegative integers
\[
    m_{\sfR_1,\sfR_2,\sfR}
    \qquad
    \text{for all }
    (\sfR_1,\sfR_2,\sfR)
    \in \calS_1\times\calS_2\times\calS
\]
such that:
\begin{itemize}
    \item $m_{\sfR_1,\sfR_2,\sfR}=0$ whenever
    $(\sfR_1,\sfR_2,\sfR)$ is not join-compatible;

    \item for every $\sfR_1\in\calS_1$,
    \[
        \sum_{\sfR_2\in\calS_2}
        \sum_{\sfR\in\calS}
        m_{\sfR_1,\sfR_2,\sfR}
        =\nu_1(\sfR_1);
    \]

    \item for every $\sfR_2\in\calS_2$,
    \[
        \sum_{\sfR_1\in\calS_1}
        \sum_{\sfR\in\calS}
        m_{\sfR_1,\sfR_2,\sfR}
        =\nu_2(\sfR_2);
    \]

    \item for every $\sfR\in\calS$,
    \[
        \sum_{\sfR_1\in\calS_1}
        \sum_{\sfR_2\in\calS_2}
        m_{\sfR_1,\sfR_2,\sfR}
        =\nu(\sfR).
    \]
\end{itemize}

Intuitively, $m_{\sfR_1,\sfR_2,\sfR}$ counts how many
occurrences of $\sfR_1$ and $\sfR_2$ are paired to produce
occurrences of $\sfR$.
These conditions allow different occurrences of the same trace
in one child to be paired with different traces in the other
child.
We call the last three conditions above the \emph{marginal constraints} of the join-compatibility.

We denote by $J(\sfS)$ the family of pairs
$(\sfS_1,\sfS_2)$ with which $\sfS$ is join-compatible.
Summing the above equalities shows that every
$(\sfS_1,\sfS_2)\in J(\sfS)$ satisfies
\[
    \ord(\sfS_1)=\ord(\sfS_2)=\ord(\sfS)
    =
    \sum_{\sfR_1\in\calS_1}
    \sum_{\sfR_2\in\calS_2}
    \sum_{\sfR\in\calS}
    m_{\sfR_1,\sfR_2,\sfR}.
\]

\begin{lemma}\label{lem:dp-join-node}
    If $t$ is a join node with children $t_1,t_2$
    and $\sfS$ is a state of $t$, then
    \[
        \dd_t[\sfS]
        =
        \min_{(\sfS_1,\sfS_2)\in J(\sfS)}
        \left\{
            \dd_{t_1}[\sfS_1]
            +\dd_{t_2}[\sfS_2]
            -\ord(\sfS)
        \right\}.
    \]
\end{lemma}

\begin{proof}
To simplify the coming argument, we regard each partial solution as a multiset of subtree-traces:
a trace of multiplicity $r$ has $r$ occurrences.
Under this interpretation, projection retains one occurrence for each nonempty projected trace, and aggregates occurrences with the same projection.

We first explain how compatible occurrences can be merged.
See Figure~\ref{fig:DP:merging} for an illustration.
Let $(\sfR_1,\sfR_2,\sfR)$ be a join-compatible triple,
and let $\sfQ_1$ and $\sfQ_2$ be subtree-traces at $t_1$
and $t_2$ projecting to $\sfR_1$ and $\sfR_2$, respectively.
Their subpaths have the same intersection with $X_t$,
with a common ordering and orientation.
Take the union of their subpaths, identifying their common
portions in $X_t$.

For each link, at most one child trace supplies a connection
through the interior of its subtree.
Consequently, this union consists of vertex-disjoint paths,
ordered according to the common ordering on the bag.
Indeed, the interiors of the two child subtrees are disjoint,
and the compatibility conditions prevent two interior
connections from occupying the same link.
A link remains unconnected precisely when it is a top link
in both children, which is precisely when it is a top link
in $\sfR$.

The endpoints of the merged trace are chosen as follows.
An endpoint represented by a vertex of $X_t$ is retained.
An endpoint whose profile in $\sfR$ is negative is taken
from the unique child in which the corresponding profile
is negative; in that child, the subtree-trace specifies
the actual endpoint.
An endpoint whose profile in $\sfR$ is positive is represented
by that positive profile.

\begin{figure}
    \centering
    \includegraphics[scale=1.2]{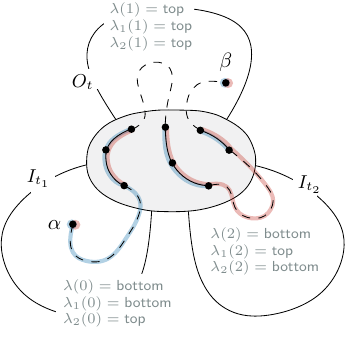}
    \caption{The merging of join-compatible partial solutions.}
    \label{fig:DP:merging}
\end{figure}

We verify the distance conditions for the resulting tuple.
Consider its first endpoint.
If it is represented by a positive profile, the 
distance function on either child subtree is the same as
in the parent, since the bags and the distance coordinates
of the profiles coincide.
If the endpoint is a vertex of $X_t$, the same exact distance
function is used throughout.
Finally, suppose that the endpoint is a vertex $a$ in
$V(G_{t_1})\setminus X_t$.
The first child uses distances from $a$, while the second
uses its positive profile.
By Lemma~\ref{lem:exact-distance}, for every
$v\in V(G_{t_2})$,
\[
d(a,v)
    =
    \delta_{t_2}\bigl(\prof_{t_2}(a),v\bigr).
\]
Thus the distance functions agree exactly on the second child.
The case where $a$ belongs to the interior of the second
child is symmetric.

It follows that distances from the first endpoint increase
by one along every oriented edge of the merged subpaths.
Between consecutive subpaths, the required strict inequality
follows from their boundary vertices in the bag and the
distance conditions of $\sfR$.
The same argument, with the orientation reversed, applies
to the other endpoint.
The explicit endpoints remain attached to the appropriate
ends of the resulting subpaths.
Moreover, every resulting subpath intersects $X_t$.
Hence the merged tuple is a subtree-trace at $t$,
and its projection is exactly $\sfR$.

\medskip

We now prove the recurrence.
First, let $(\sfS_1,\sfS_2)\in J(\sfS)$.
If either child value is infinite, the inequality
\[
    \dd_t[\sfS]
    \leq
    \dd_{t_1}[\sfS_1]+\dd_{t_2}[\sfS_2]-\ord(\sfS)
\]
is immediate.
Otherwise, choose minimum-order partial solutions
$\sfP_1$ and $\sfP_2$ realizing $\sfS_1$ and $\sfS_2$,
and fix nonnegative integers
$m_{\sfR_1,\sfR_2,\sfR}$ witnessing join-compatibility.

For every $\sfR_1\in\calS_1$, there are exactly
$\nu_1(\sfR_1)$ occurrences in $\sfP_1$ projecting to
$\sfR_1$.
The first marginal constraint allows us to distribute these
occurrences among the pairs $(\sfR_2,\sfR)$, allocating
$m_{\sfR_1,\sfR_2,\sfR}$ occurrences to each pair.
Do the analogous distribution in $\sfP_2$ using the
second marginal constraint.
For each triple with positive multiplicity, pair the two
allocated collections arbitrarily and merge corresponding
occurrences by the construction above.

Retain all complete trace occurrences from both children.
Together with the merged occurrences, they form a partial
solution $\sfP$ at $t$.
Indeed, merging preserves every vertex covered by either
occurrence, and
$V(G_t)=V(G_{t_1})\cup V(G_{t_2})$.
The third marginal constraint ensures that exactly
$\nu(\sfR)$ merged occurrences project to each
$\sfR\in\calS$.
Thus $\sfP$ realizes $\sfS$.

Each merge replaces two occurrences by one, and the number
of merges is
\[
    \sum_{\sfR_1,\sfR_2,\sfR}
    m_{\sfR_1,\sfR_2,\sfR}
    =\ord(\sfS).
\]
Therefore,
\[
    \dd_t[\sfS]
    \leq
    \dd_{t_1}[\sfS_1]+\dd_{t_2}[\sfS_2]-\ord(\sfS).
\]
Taking the minimum over $J(\sfS)$ proves the first inequality.

\medskip

For the converse, assume that $\dd_t[\sfS]$ is finite,
and let $\sfP$ be a minimum-order partial solution
realizing $\sfS$.
Restrict each of its occurrences to each child subtree,
discarding empty restrictions.
When an explicit endpoint lies in the interior of the
other child, replace it by its positive profile.
Positive endpoint profiles retain their distance coordinates,
since the bags coincide.
Denote the resulting multisets by $\sfP_1$ and $\sfP_2$.

These restrictions are valid subtree-traces.
To see this, intersect each subpath with the corresponding
child subtree and retain the inherited ordering.
Every cut occurs at the common separator $X_t$.
The distance conditions are inherited from the parent:
replacing an endpoint in the other child by its profile
preserves its distance function,
by Lemma~\ref{lem:exact-distance}.
The endpoint attachment conditions are preserved as well.
Since restriction preserves coverage of the corresponding
child subtree, $\sfP_1$ and $\sfP_2$ are partial solutions.
Let $\sfS_1$ and $\sfS_2$ be their projections.

Every occurrence of $\sfP$ meeting $X_t$ has a nonempty
restriction to both children.
Its two child projections and its parent projection form
a join-compatible triple.
Indeed, their bag subpaths coincide, and each endpoint
outside the bag lies either in one child interior or
outside the parent subtree.
Likewise, each connection through the parent interior
lies entirely in one child interior, since $X_t$ separates
the two interiors.

Define $m_{\sfR_1,\sfR_2,\sfR}$ to count the occurrences
of $\sfP$ having these three projections.
These integers satisfy all three marginal constraints,
so $(\sfS_1,\sfS_2)\in J(\sfS)$.

Every complete occurrence of $\sfP$ lies entirely in one
child interior and is counted once in $\sfP_1,\sfP_2$.
Every occurrence meeting $X_t$ is counted twice.
There are exactly $\ord(\sfS)$ occurrences of the latter
kind. Consequently,
\[
    \ord(\sfP_1)
    +\ord(\sfP_2)
    =
    \ord(\sfP)+\ord(\sfS).
\]
By minimality of the child values,
\[
\begin{aligned}
    \dd_t[\sfS]
    &=
    \ord(\sfP_1)
    +\ord(\sfP_2)-\ord(\sfS)\\
    &\geq
    \dd_{t_1}[\sfS_1]+\dd_{t_2}[\sfS_2]-\ord(\sfS)\\
    &\geq
    \min_{(\sfS'_1,\sfS'_2)\in J(\sfS)}
    \left\{
        \dd_{t_1}[\sfS'_1]
        +\dd_{t_2}[\sfS'_2]
        -\ord(\sfS)
    \right\}.
\end{aligned}
\]
If $\dd_t[\sfS]$ is infinite, this inequality is immediate.
Together, the two inequalities prove the recurrence,
with the convention that the minimum over an empty family
is $+\infty$.
\end{proof}

\subsection{Correctness and complexity of the algorithm}

We now prove the correctness and complexity of the algorithm.

\begin{lemma}\label{lem:dp-root-node}
If $r$ is the root, the minimum size of an isometric path cover of $G$
is $\dd_r[\sfS_0]$.
\end{lemma}

\begin{proof}
Here $X_r=\emptyset$ and $G_r=G$. 
Every subtree-trace is complete and hence consists of one isometric path with explicit endpoints.
Thus a realization of $\sfS_0$ is a set of isometric paths covering $G$ with no multiplicity. 
Its order thus coincides with its cardinal.
We conclude to the claimed statement by minimality over the realizations of $\sfS_0$.
\end{proof}

The algorithm computes the values of the states bottom-up, using
Lemmas~\ref{lem:dp-leaf-node}, \ref{lem:dp-introduce-node},
\ref{lem:dp-forget-node}, and \ref{lem:dp-join-node}, and outputs
$\dd_r[\sfS_0]$. 
Lemma~\ref{lem:dp-root-node} proves correctness.
An optimal cover can also be recovered by storing a minimizing transition
for every finite entry, including the allocation integers at joins, and
following the constructive proofs of the bottom-up computation of the values of the states.

We now evaluate the number of states.

\begin{lemma}\label{lem:number-of-states}
Let $G$ be a connected graph on $n\geq 2$ vertices
with diameter $\diam\geq 2$, and let $t$ be a node whose bag has size $q\geq 1$.
The number of distinct bag-traces at $t$ is at most
\[
    B_q := \bigl(q+2\diam^q\bigr)^2\,2^{3q+1}.
\]
Consequently, the number of states at $t$ with
multiplicities in $\{1,\dots,n\}$ is at most
\[
    (n+1)^{B_q}
    \leq
    n^{\diam^{O(q)}}.
\]
\end{lemma}

\begin{proof}
We first count bag-traces.
Each endpoint is either a vertex of $X_t$ or a signed
distance profile.
For a vertex $a\notin X_t$, every coordinate of its
profile belongs to $\{1,\dots,\diam\}$.
There are therefore at most $2\diam^q$ signed profiles,
and at most $(q+2\diam^q)^2$ ordered choices of endpoints.

Fix such a choice $(\alpha,\beta)$.
The distance conditions of a bag-trace imply that
the distances from $\alpha,\beta$ to the subpaths is strictly increasing along the entire
ordered sequence of its selected bag vertices:
it increases by one within each subpath, and strictly
increases between consecutive subpaths.
Consequently, once the set $U\subseteq X_t$ of selected vertices is fixed, their order is uniquely determined.
There are at most $2^q$ choices for $U$.
Then, partitioning the order on $U$ into consecutive parts determines the ordered subpaths, and there are at most $2^q$ such partitions.

If a partition has $k$ subpaths, there are at most
$k+1\leq q+1$ links: those between consecutive subpaths
and those incident with the two endpoints.
Each link has at most two possible values.
Thus there are at most $2^{q+1}$ linking functions.

Multiplying these bounds, the number of bag-traces
is at most
\[
    (q+2\diam^q)^2\cdot 2^q\cdot 2^q\cdot 2^{q+1}
    =B_q.
\]

A state with positive multiplicities at most $n$ can
be encoded by assigning a value in $\{0,\dots,n\}$
to every possible bag-trace, where zero means that
the trace is absent.
Hence there are at most $(n+1)^{B_q}$ states.
Since $\diam\geq 2$ and $q\geq 1$, $B_q=\diam^{O(q)}$ and the second bound follows.
\end{proof}

We now evaluate the computation of the transitions.
All distances and actual profile classes can be computed in polynomial time per decomposition node, after an all-pairs shortest-path computation.
Bag-traces can be enumerated using the encoding in the proof: choose endpoints, a nonempty selected vertex set, a partition of its forced
order, and genuine link labels, then test the defining conditions.
These tests use only graph adjacency, distances, profiles, and bag data.

At an introduce node, admissibility, restriction, and the additional cost are computed directly from each state. 
At a forget node, we enumerate the child states, test admissibility, and update the entry for their restriction.
At a join node, we enumerate parent and child state triples, compute the marginal constraint and test compatibility.
The number of such triples is at most $\diam^{O(w)}$, and the number of integer values for $m_{\sfR_1,\sfR_2,\sfR}$ is $n^{\diam^{O(w)}}$.
These operations take $n^{\diam^{O(w)}}$ time, times a polynomial in the size of $G$. 

\begin{theorem}\label{thm:ipc-treewidth-diameter}
Given a connected $n$-vertex graph of diameter $\diam$ and a nice tree
decomposition of width $w$ with $N$ nodes, a minimum isometric path cover
can be computed in time
\[
 N\cdot n^{\diam^{O(w)}}.
\]
Consequently, \SoPC{} is in \XP{} parameterized by treewidth plus the diameter of the graph;
without a supplied decomposition, the running time is
$n^{\diam^{O(\operatorname{tw})}}$ for $n\geq2$.
\end{theorem}

\begin{proof}
Correctness and the bound for a supplied nice decomposition follow from
the preceding arguments. For the second statement, Korhonen's algorithm
\cite{DBLP:conf/focs/Korhonen21} produces a decomposition of width at most
$2k+1$ in time $2^{O(k)}n$ when $k=\operatorname{tw}(G)$.
Trying successive values of $k$ until a decomposition is returned has the
same asymptotic bound and may terminate earlier.
A decomposition of width $w$ with linearly many bags can be put in nice
form with $O((w+1)n)$ nodes and empty root and leaves by the standard
normalization~\cite{bodlaender2013fine,kloks1994treewidth}.
The parameter-dependent preprocessing and the number of nodes are absorbed
by $n^{\diam^{O(k)}}$.
\end{proof}


\section{\sloppy \textsc{\SoPC} W[1]-Hardness with Respect to Pathwidth on Graphs of Diameter 7} \label{sec:cover-W1h-tw-diam}

\restatethmIPCWhtwdiam*

The reduction is from \textsc{Unary Bin Packing}: given positive integer
item sizes and a bin capacity $b$, encoded in unary, decide whether the
items fit into at most $k$ bins of capacity $b$. This problem admits a
pseudopolynomial dynamic program with running time $n b^{O(k)}$, but is
\W[1]-hard parameterized by $k$ \cite{JansenKMS10}.
We use the following variant.

\defproblem{\textsc{Special Unary Bin Packing}}{$n$ items of even sizes
$t_1,\ldots,t_n\ge4$, a positive capacity $b$, all encoded in unary, and a positive
integer $k$, with $\sum_i t_i=kb$.}{Can we partition the items into $k$
bins, each of total size exactly $b$?}

\begin{lemma}
\textsc{Special Unary Bin Packing} is \W[1]-hard parameterized by $k$.
\end{lemma}
\begin{proof}
Start with an instance of \textsc{Unary Bin Packing}, and let $T$ be its
total item size. Instances with an item larger than $b$ or $T>kb$ are
immediate no-instances; if $k\ge n$ and every item fits, the answer is yes.
Map these cases to fixed instances with the corresponding answer.
Otherwise, add $kb-T$ items of size one, and multiply all sizes and $b$
by four. The added items fill precisely the unused capacity of any
packing. Conversely, deleting them from a solution gives a packing of
the original instance. Outside the fixed cases, the transformation has polynomial size,
preserves $k$, and enforces the stated size and total-sum conditions. 
\end{proof}

\subsection{Overview of the reduction}

In the reduction, we construct a graph $G$ that is shown in Figure \ref{fig:AS-full_reduction}. Among the vertices of the graph there are $n$ vertices $a_1,\dots, a_n$ that correspond to the items, and $k$ vertices $b_1,\dots, b_k$ that correspond to the bins. Each vertex $a_i$ has $t_i$ leaves that encode the size $t_i$ of this item, and each vertex $b_j$ has $b$ leaves that encode the capacity of the corresponding bin. 
Then, each pair $(a_i,b_j)$ is connected by a structure called an \emph{anti-split} gadget that essentially consists of $t_i$ ``parallel'' paths plus some other edges and vertices. If the $i$-th item is placed into the $j$-th bin, then we create $t_i$ paths from the leaves of $a_i$ to $t_i$ distinct leaves of $b_j$ that traverse the corresponding anti-split gadget. A potential problem here is that, in order to cover the graph, one could be tempted to ``split'' the items by having some paths starting from the same item but finishing in different bins. Our anti-split gadget prevents such a split by inducing a ``penalty'' for each gadget that is only partially traversed. 

\paragraph{Organisation.} In Section \ref{sec:anti-split-gadget}, we introduce our anti-split gadget, and describe the reduction in Section \ref{sec:AS:construction}, where we give some properties of the graph constructed. Finally, we prove the equivalence of the instances in Sections \ref{section:AS:forward} and \ref{section:AS:backward}. 

\subsection{The anti-split gadget}\label{sec:anti-split-gadget}

An anti-split gadget has two \emph{connection vertices} $a,b$ and
$8t$ \emph{internal vertices}, where its capacity $t\ge4$ is even.
Figure~\ref{fig:AS_intro} shows the case $t=4$.
The repeated pattern is cyclic: indices below are taken modulo $t$.
For each $r\in\{1,\ldots,t\}$, introduce the eight vertices
$x_r,p_r,q_r,y_r,c_r,d_r,e_r,f_r$, and add the paths
\[
 a,x_r,p_r,q_r,y_r,b;\qquad
 x_r,c_r,d_r,y_{r-1};\qquad d_r,e_r,f_r,c_{r+1}.
\]
These are all the gadget edges. In particular, lines crossing without a
marked vertex in the figures do not create adjacencies. The vertices
$x_r$ and $y_r$ will belong to $X$ and $Y$, respectively. In the final
graph they also have neighbors $w_1$ and $w_6$ outside the gadget.

We will use the following two types of internal paths, shown in green
and blue in Figure~\ref{fig:AS_intro}:
\begin{align*}
 \text{green at }r:\quad&
 f_{r-1},c_r,x_r,p_r,q_r,y_r,d_{r+1},e_{r+1};\\
 \text{blue at }r:\quad&
 e_r,d_r,c_r,f_{r-1},e_{r-1},d_{r-1},c_{r-1},f_{r-2}.
\end{align*}
Both have length seven. All $t$ green paths partition the internal
vertices. After the $t$ vertical paths $x_r,p_r,q_r,y_r$ are used,
the blue paths with even indices partition the remaining vertices.
Their isometry in the final graph is verified below.

\begin{figure}[h]
    \centering
    \includegraphics[page = 1, width=0.7\linewidth]{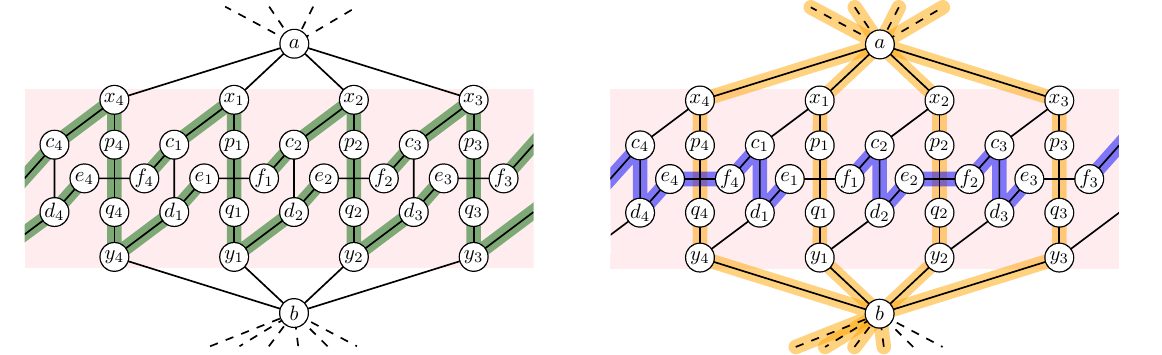}
    \caption{Two copies of an anti-split gadget of capacity $t=4$. The left side shows the optimal way to cover the internal vertices with $t$ isometric paths (green). We refer to this type of isometric paths as \emph{internal green paths}. The right side shows a situation where the gadget is traversed by $t$ isometric paths from $a$ to $b$. In this case, the optimal way to cover the remaining vertices is to use $t/2$ isometric paths (blue). We refer to this type of isometric paths as \emph{internal blue paths}}
    \label{fig:AS_intro}
\end{figure}

\subsection{Construction of the graph}\label{sec:AS:construction}

Let $(\{t_1,\dots, t_n\},k,b)$ be an instance of \textsc{Special Unary Bin Packing} where $t_1,\dots, t_n$ denote the sizes of the items. Notice that $\sum_i t_i=kb$. 

We construct a graph $G=(V,E)$ as follows. Figure \ref{fig:AS-full_reduction} illustrates the construction. 

\begin{figure}
    \centering
    \includegraphics[page=2,width=1\linewidth]{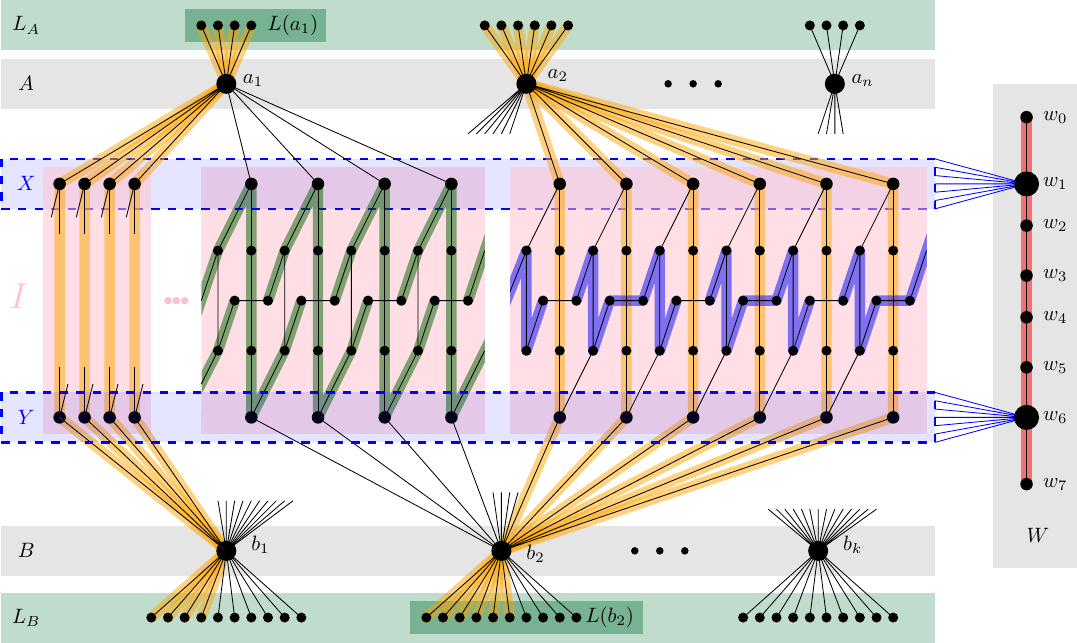}
    \caption{The graph of the reduction from \textsc{Special Unary Bin Packing}. Orange paths are \emph{encoding paths} that indicate in which bin each item is packed. The path $W$ bounds the diameter and is forced to appear in every cover within the budget. Internal vertices (pink area) that are not covered by encoding paths are covered either by \emph{internal green paths} (middle gadget) or \emph{internal blue paths} (rightmost gadget).}
    \label{fig:AS-full_reduction}
\end{figure}

\begin{itemize}
    \item For each $i\in [n]$, add a vertex $a_i$ and $t_i$ leaf vertices $L(a_i)$ adjacent to $a_i$. 
\end{itemize}
\begin{itemize}
    \item For each $j\in [k]$, add a vertex $b_j$ and $b$ leaf vertices $L(b_j)$ adjacent to $b_j$.
\end{itemize}
 We denote $A=\{a_1,\dots,a_n\}$, $B=\{b_1,\dots, b_k\}$, $L_A=\cup_{i\in[n]} L(a_i)$, $L_B=\cup_{j\in[k]} L(b_j)$, and $L=L_A\cup L_B$. 
\begin{itemize}
    \item For each pair $(i,j)\in [n]\times [k]$, add an \emph{anti-split gadget} of capacity $t_i$ whose connection vertices are $a_i$ and $b_j$. 
\end{itemize}

Let $I$ denote the set of all internal vertices of these gadgets. 
Let $X\subset I$ (resp. $Y\subset I$) denote the set of internal vertices that are adjacent to $a_i$ for some $i\in [n]$ (resp. that are adjacent to $b_j$ for some $j\in [k]$, respectively). 

\begin{itemize}
    \item Finally, add a path $W=\{w_0,\dots w_7\}$ of length $7$, then connect $w_1$ to all vertices in $X$ and $w_6$ to all vertices in $Y$. 
\end{itemize}
This completes the construction of $G$. Set
$K=1+kb/2+k^2b$ which will be the cardinality of the optmial isometric path cover when the input is an \textsc{Yes}-instance.. 

\begin{lemma}
    $G$ has $n+k+2kb+8k^2b+8$ vertices.
\end{lemma}
\begin{proof}
    The first two terms correspond to sets $A$ and $B$ respectively. The third term $2kb$ is the size of $L=L_A\cup L_B$. Then we have $|I|=8k^2b$, since a gadget associated to $a_i$ has $8t_i$ internal vertices and for each $i\in [n]$, there are $k$ gadgets incident to $a_i$. Thus, $|I|=\sum_{i\in [n]}8kt_i=8k(kb)=8k^2b$. Finally, the last term is the size of $W$. 
\end{proof}

\begin{lemma}
    $G$ has diameter $7$. 
    \label{claim:AS_diameter}
\end{lemma}

\begin{figure}
    \centering
    \includegraphics[page= 5,width=\linewidth]{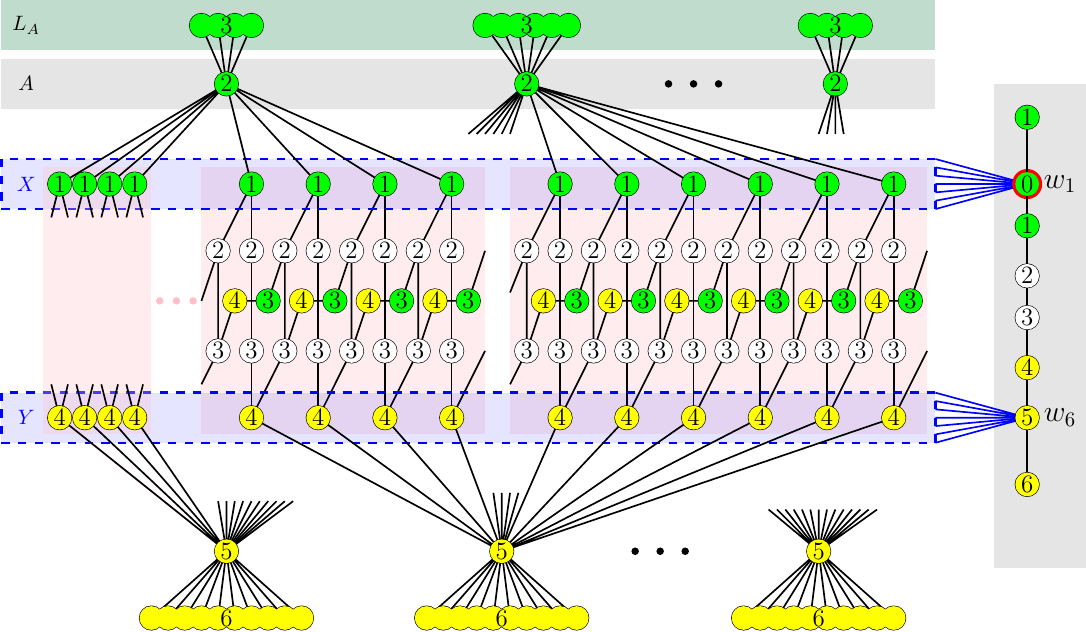}
    \caption{Distances from $w_1$. Yellow vertices are at distance at least $4$ from $w_1$ while green vertices are at distance at least $4$ from $w_6$. }
    \label{fig:AS-dist-from-w}
\end{figure}

\begin{figure}
    \centering
    \includegraphics[page= 3,width=\linewidth]{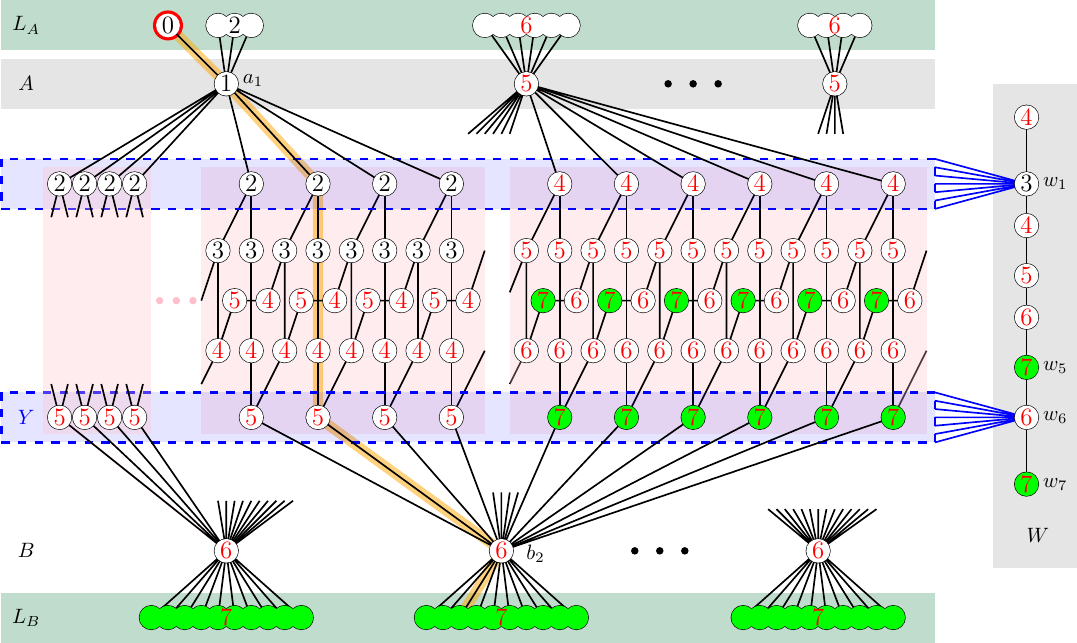}
    \caption{Distances from a vertex in $L_A$. The orange path is an encoding path from that vertex. Green vertices are at distance $7$ from the highlighted leaf. }
    \label{fig:AS-dist-from-leaf}
\end{figure}

\begin{figure}
    \centering
    \includegraphics[page= 4,width=\linewidth]{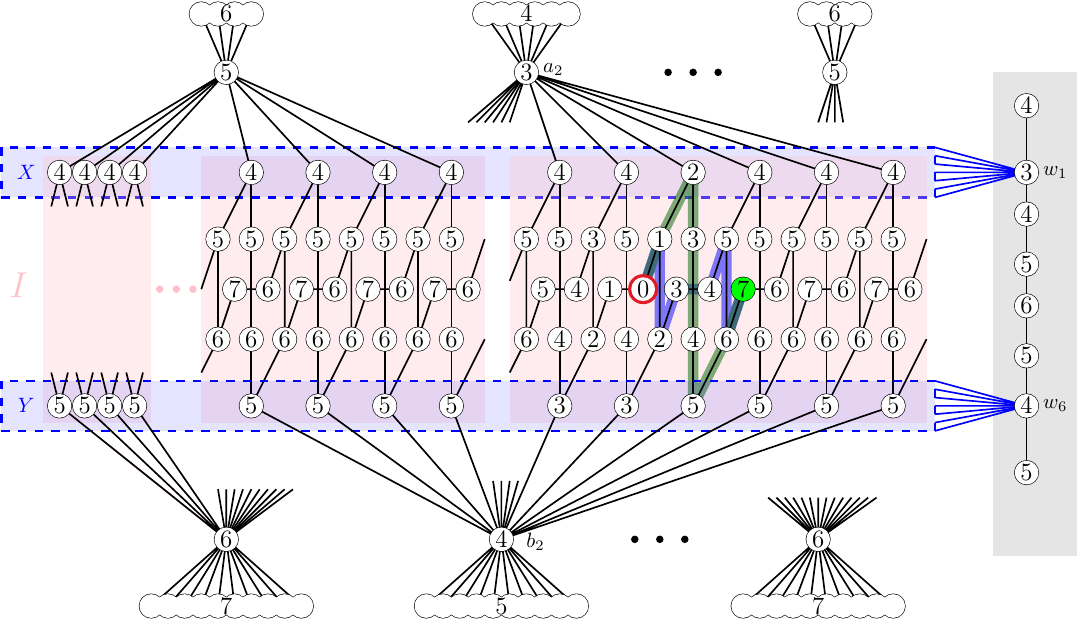}
    \caption{Distances from an internal vertex (red node). The highlighted green vertex is the only endpoint reachable from the red vertex by a length-seven isometric path contained in this gadget. There are two such paths, shown in green and blue. }
    \label{fig:AS-dist-from-z}
\end{figure}

\begin{proof}
The proof uses Figures~\ref{fig:AS-dist-from-w},
\ref{fig:AS-dist-from-leaf}, and~\ref{fig:AS-dist-from-z}, which show
the distances from $w_1$, an item leaf, and an internal vertex,
respectively. These distances follow from the edges of the construction.

Suppose, for a contradiction, that two vertices $v,v'$ satisfy
$d(v,v')\ge8$. Let $H=\{u\in V:d(w_1,u)\ge4\}$ and
$H'=\{u\in V:d(w_6,u)\ge4\}$. These sets are shown in yellow and
green in Figure~\ref{fig:AS-dist-from-w}. If neither endpoint belonged
to $H$, the route through $w_1$ would have length at most six.
Thus $\{v,v'\}$ meets $H$, and the same argument using $w_6$ shows
that it meets $H'$. Since $H\cap H'=\emptyset$, we may assume that
$v\in H$ and $v'\in H'$.

We now consider the possible locations of $v'$. The set $H'$ consists
of $A\cup L_A\cup X$, the vertices $w_0,w_1,w_2$, and the vertices
$f_r$ of the gadgets.
\begin{itemize}
 \item If $v'=w_1$, Figure~\ref{fig:AS-dist-from-w} shows that every
 vertex is at distance at most six from $v'$.
 \item If $v'\in\{w_0,w_2\}\cup X$, then $v'$ is adjacent to $w_1$,
 so every vertex is at distance at most seven from $v'$.
 \item If $v'\in L_A$, Figure~\ref{fig:AS-dist-from-leaf} shows that
 every vertex is at distance at most seven from $v'$. The same bound
 holds for any choice of item leaf.
 \item If $v'\in A$, choose a leaf adjacent to $v'$. Distances from
 $v'$ to every other vertex are one less than the distances from that
 leaf. Thus every vertex is at distance at most six from $v'$.
 \item The remaining possibility is that $v'$ has type $f_r$ in a
 gadget. This is the type of the red vertex in
 Figure~\ref{fig:AS-dist-from-z}, which shows that every vertex is
 at distance at most seven. Cyclic symmetry gives the same conclusion
 for every index $r$, in any gadget.
\end{itemize}
Each case contradicts $d(v,v')\ge8$, so the diameter is at most seven.
Finally, $d(w_0,w_7)=7$: a route from $w_1$ to $w_6$ has at least
five edges, whether it follows $W$ or passes through the gadgets.
Therefore the diameter is exactly seven.
\end{proof}

\begin{lemma}
\label{lemma:AS:pathwidth}
The graph $G$ has pathwidth at most $k+32$.
\end{lemma}
\begin{proof}
We construct a path decomposition by keeping $B\cup W$ in every
bag and processing the remaining vertices one item at a time. First,
introduce and forget the leaves of $B$ one at a time. When processing
item $i$, retain $a_i$, process its leaves in the same way, and then
process its $k$ gadgets consecutively.

It remains to describe how to process one gadget without keeping all
its vertices in a bag. Let $C_r$ consist of its eight internal vertices
with index $r$. Use the following bags, in the indicated order:
\[
 B\cup W\cup\{a_i\}\cup C_1\cup C_r\cup C_{r+1},
 \qquad r=1,\ldots,t_i-1.
\]
Every internal edge lies within one cell or joins two cyclically
consecutive cells. The cells $C_r,C_{r+1}$ cover the consecutive-cell
edges, while retaining $C_1$ throughout covers the edges joining the
last cell to the first. All edges incident with $a_i$, $b_j$, $w_1$,
or $w_6$ are covered because these vertices remain in the bags.
Moreover, every internal vertex occurs in a consecutive sequence of
bags. We can therefore concatenate these sequences for all gadgets
and items to obtain a path decomposition of $G$.

The largest bag contains at most $k+8$ vertices of $B\cup W$, the
vertex $a_i$, and three cells of eight vertices. Its size is at most
$k+33$, giving pathwidth at most $k+32$.
\end{proof}

In Sections \ref{section:AS:forward} and \ref{section:AS:backward}, we prove the following equivalence. 

\begin{lemma}
The instance of \textsc{Special Unary Bin Packing} has a solution if and only if $G$ can be covered by at most $K=1 + \frac{kb}{2} + k^2b$ isometric paths. 
\label{lemma:AS:equivalence}
\end{lemma}

Thanks to this lemma we can prove Theorem \ref{thm:cover-W1h-tw-diam}. 

\begin{proof}[Proof of Theorem \ref{thm:cover-W1h-tw-diam}]
The construction has polynomial size in the unary input, diameter
seven, and pathwidth at most $k+32$ by Lemma~\ref{lemma:AS:pathwidth}.
Lemma~\ref{lemma:AS:equivalence} establishes equivalence of the instances.
This is a parameterized reduction from
\textsc{Special Unary Bin Packing}, proving the claimed pathwidth hardness.
\end{proof}

\subsection{Forward direction}\label{section:AS:forward}

Suppose the instance of \textsc{Special Unary Bin Packing} has a solution,
with item $i$ packed into bin $j(i)$. The green and blue paths listed
above are isometric in $G$: tracing their internal routes gives distance
seven between their endpoints, and a route leaving the gadget meets one
of $a_i,b_j,w_1,w_6$, whose distances to these endpoints sum to seven.
These distances are illustrated in Figure~\ref{fig:AS-dist-from-z}. 

We construct a cover $\mathcal{P}$ of $G$ with $kb + \frac{kb}{2} + (k^2b-kb) + 1$ isometric paths as follows. 
\begin{itemize}
    \item add a path of length $7$ that covers $W$. Figure \ref{fig:AS-dist-from-w} shows that this is an isometric path.
    \item For each $i\in [n]$, add $t_i$ \emph{encoding} paths, each one starting in a different leaf of $L(a_i)$ and finishing in a different leaf of $L(b_{j(i)})$. Such a family is shown on Figure \ref{fig:AS-full_reduction} by the orange paths (on the figure we have $j(2)=2$). These encoding paths go through the anti-split gadget with connection vertices $a_i$ and $b_{j(i)}$, and each uses a different vertical path $x_r,p_r,q_r,y_r$. Assign bin leaves bijectively to the incoming paths; the packing guarantees exactly $b$ paths per bin. Figure \ref{fig:AS-dist-from-leaf} shows that encoding paths are isometric. There are $\sum_{i\in [n]}t_i=kb$ encoding paths.
    \item For each $i\in [n]$, add the $t_i/2$ \emph{internal blue} paths with even indices that cover the uncovered vertices of the anti-split gadget connecting $a_i$ and $b_{j(i)}$. These paths are shown on Figure \ref{fig:AS-full_reduction} (on the gadget between $a_2$ and $b_2$), and on Figure \ref{fig:AS_intro} (right). Figure \ref{fig:AS-dist-from-z} shows that these paths are isometric. In total, $\mathcal{P}$ uses $\sum_{i\in [n]}t_i/2=kb/2$ internal blue paths. 
    \item For each $(i,j)\in [n]\times [k]$ such that $j\neq j(i)$, add $t_i$ \emph{internal green paths} to cover the gadget between $a_i$ and $b_j$. These paths are shown on Figure \ref{fig:AS-full_reduction} (in the gadget between $a_1$ and $b_2$), and on Figure \ref{fig:AS_intro} (left). Figure \ref{fig:AS-dist-from-z} shows that these paths are isometric. In total, the number of internal green paths is $\sum_{i\in [n]}\sum_{j\in [k], j\neq j(i)} t_i=(k-1)\sum_{i\in [n]} t_i=(k-1)kb$. 
\end{itemize}
This defines $\mathcal{P}$. Observe that, $\mathcal{P}$ has size $1+kb+kb/2+(k-1)kb=1 +{kb}/{2} + k^2b$. This concludes the proof of the forward direction. Notice that all paths in $\mathcal{P}$ have length exactly $7$, the diameter of the graph, and that the only vertices of the graph that are covered by several paths are vertices of $A\cup B$. 

\subsection{Backward direction}\label{section:AS:backward}
Fix an isometric path cover $\mathcal P$ of $G$ with at most
$K=1+kb/2+k^2b$ paths. We first establish the properties that will
allow us to recover a packing.

\begin{lemma}
\label{lemma:AS:cover-structure}
The cover $\mathcal P$ has the following properties.
\begin{enumerate}
 \item Every path has length seven.
 \item Every vertex outside $A\cup B$ is covered exactly once.
 \item Apart from $W$, every path either consists entirely of internal
 vertices of one anti-split gadget, or joins a leaf of $a_i$ to a leaf
 of $b_j$ through a vertical path $x_r,p_r,q_r,y_r$ of their gadget,
 as the orange paths in Figure~\ref{fig:AS-full_reduction}.
 In particular, every path meeting $A\cup B\cup L$ has this latter form.
\end{enumerate}
\end{lemma}
\begin{proof}
Here $|P|$ denotes the number of vertices of a path $P$.
Let $\mathcal P_{\le3}=\{P\in\mathcal P:|P|\le3\}$ and
$\mathcal P_{\ge4}=\mathcal P\setminus\mathcal P_{\le3}$.
For each vertex $v$, let $c(v,\mathcal P)$ be the number of paths
of $\mathcal P$ containing $v$. For a vertex set $V'$, put
$c(V',\mathcal P)=\sum_{v\in V'}c(v,\mathcal P)$.
Thus $c(V',\mathcal P)$ counts coverage with multiplicity, and double
counting gives
\[
 c(V,\mathcal P)=\sum_{P\in\mathcal P}|P|.
\]

We first give a lower bound on the coverage of $A\cup B$.
Let $L_{\ge4}\subseteq L$ be the leaves not covered by any path
in $\mathcal P_{\le3}$. There are $2kb$ leaves in $L$, and a path can
cover at most two leaves. Hence
\[
 |L_{\ge4}|\ge2kb-2|\mathcal P_{\le3}|.
\]
Every leaf of $L_{\ge4}$ must be covered by a path in
$\mathcal P_{\ge4}$. Now consider any such path $P$. If it contains
a leaf of $L$, it also contains its neighbor in $A\cup B$.
Furthermore, it cannot contain two leaves with the same neighbor:
those leaves would be its endpoints and the path would have only
three vertices. Therefore
$|P\cap(A\cup B)|\ge|P\cap L|$, and summing gives
\begin{align}
 c(A\cup B,\mathcal P)
 &\ge c(A\cup B,\mathcal P_{\ge4})\nonumber\\
 &=\sum_{P\in\mathcal P_{\ge4}}|P\cap(A\cup B)|\nonumber\\
 &\ge\sum_{P\in\mathcal P_{\ge4}}|P\cap L|\label{eq:comp_A_LA}\\
 &=c(L,\mathcal P_{\ge4})\nonumber\\
 &\ge |L_{\ge4}|\ge2kb-2|\mathcal P_{\le3}|.\nonumber
\end{align}

We now compare the coverage required by the graph with the coverage
provided by the paths. On the one hand, every vertex outside $A\cup B$
must be covered at least once. Using the vertex count and
the preceding bound, we obtain
\begin{align}
 c(V,\mathcal P)
 &=c(V\setminus(A\cup B),\mathcal P)+c(A\cup B,\mathcal P)\nonumber\\
 &\ge |V\setminus(A\cup B)|+c(A\cup B,\mathcal P)
       \label{eq:coverage_rest_a}\\
 &=2kb+8k^2b+8+c(A\cup B,\mathcal P)\nonumber\\
 &\ge4kb+8k^2b+8-2|\mathcal P_{\le3}|.
       \label{eq:ineq-lemma}
\end{align}
On the other hand, an isometric path has at most eight vertices by
Lemma~\ref{claim:AS_diameter}, so
\begin{align}
 c(V,\mathcal P)
 &=\sum_{P\in\mathcal P}|P|\nonumber\\
 &\le3|\mathcal P_{\le3}|+8|\mathcal P_{\ge4}|
       \label{eq:lenght-paths}\\
 &=8|\mathcal P|-5|\mathcal P_{\le3}|\nonumber\\
 &\le4kb+8k^2b+8-5|\mathcal P_{\le3}|,\nonumber
\end{align}
where the last inequality uses $|\mathcal P|\le K=1+kb/2+k^2b$.
Comparing the two bounds forces $|\mathcal P_{\le3}|=0$.
The bounds then coincide, so $|\mathcal P|=K$ and all the displayed
inequalities are equalities. In particular, equality
in~\eqref{eq:lenght-paths} means that every path has eight vertices,
and equality in~\eqref{eq:coverage_rest_a} means that no vertex outside
$A\cup B$ is covered twice. This proves the first two properties.

Moreover, equality in~\eqref{eq:ineq-lemma} gives
$c(A\cup B,\mathcal P)=2kb$, and the coverage bounds above are tight. In particular, equality in~\eqref{eq:comp_A_LA} gives
$|P\cap(A\cup B)|=|P\cap L|$ for every $P\in\mathcal P$.
Every leaf of $P\cap L$ has its neighbor on $P$, and these neighbors
are distinct because $P$ has eight vertices. They therefore account
for all vertices of $P\cap(A\cup B)$. Consequently, if $u\in A\cup B$
belongs to $P$, one of its leaves belongs to $P$ as well.

We next show that $W$ belongs to the cover. This will exclude paths
that use $w_1$ or $w_6$ to move between gadgets.
    Refer to Figure~\ref{fig:AS-dist-from-w}. 
    Let $P$ denote the isometric path that contains $w_0$, and suppose for a contradiction that $P\neq W$. Since $|P|=8$, either it uses an edge between $w_1$ and a vertex in $X$, or it uses an edge between $w_6$ and a vertex in $Y$. Notice that the second case cannot happen since $P$ is an isometric path, $d(w_0,w_6)=6$ and the distance between $w_0$ and any vertex of $Y$ is $5$. Thus, $w_2\notin P$, and let $P'\in \mathcal{P}$ be the path that contains $w_2$. Since $w_1$ is covered only once, $w_2$ is an endpoint of $P'$. Its length is seven, so we must have $w_6\in P'$ and  $w_7\notin P'$. But then, the path that contains $w_7$ must intersect $P'$ in $w_6$ which contradicts the fact that this vertex is covered only once.

We now show that a path meeting $L_A$ must end in $L_B$, and conversely.
    Refer to Figure~\ref{fig:AS-dist-from-leaf}. 
    Let $P$ be a path in $\mathcal{P}$ that intersects $L_A$. The other case is symmetric. We show that $P$ has to connect to a leaf in $L_B$ like the orange path shown in the Figure. 
    Since $|P|=8$, the other endpoint $u'$ must be one of the (green) vertices labeled with $7$ on Figure \ref{fig:AS-dist-from-leaf}. 

    Since $w_1$ and $w_6$ are already covered by the preceding argument and are covered only once, we know that $u'\notin \{w_5,w_7\}$. Then, $u'$ cannot be one of the green internal vertices that are not in $Y$ since all the corresponding isometric paths contain $w_1$ which would again contradict unique coverage. 

    Suppose now that $u'\in Y$. Then, either $P$ passes by $w_1$ which would again be a contradiction,  or passes by a vertex in $B$ ($b_2$ on the figure). This contradicts the fact proved above that a covered connection vertex must have a leaf on the same path, because $P$ would then contain a vertex $b_j$, for some $j\in[k]$, without containing a leaf in $L(b_j)$. 

    Thus, the only remaining option is $u'\in L_B$. This establishes the required endpoint pairing.

Consider such a path, with endpoints in $L(a_i)$ and $L(b_j)$.
It cannot meet another connection vertex, since that would force a
third leaf onto the path. It also avoids $W$, by unique coverage.
Its internal segment therefore stays in the gadget between $a_i$
and $b_j$. Four of its seven edges connect its endpoint leaves to
$X$ and $Y$, leaving three edges within the gadget. From the gadget
edges, this segment is either $x_r,p_r,q_r,y_r$ or
$x_r,c_r,d_r,y_{r-1}$.

We can exclude the second possibility using the vertices $p_r,q_r$.
Each of these vertices is at distance at most six from every vertex
of $G$. Indeed, their distances to $(w_1,w_6)$ are $(2,3)$ and $(3,2)$,
respectively; the distances in Figure~\ref{fig:AS-dist-from-w} give
the bound six in both cases. Neither can therefore be an endpoint of
a length-seven isometric path. Since they have degree two, the path
covering $p_r$ must contain the entire segment $x_r,p_r,q_r,y_r$.
An encoding path using $x_r,c_r,d_r,y_{r-1}$ cannot also contain
$p_r$. It would therefore share $x_r$ with the path covering $p_r$,
contradicting unique coverage. Thus every path meeting $L$
is an orange encoding path.

Finally, any path that meets $A\cup B$ also meets $L$, as proved above.
Every remaining path other than $W$ avoids both the connection vertices
and $W$, so it is contained in the internal vertices of a single gadget.
This proves the third property and completes the proof.
\end{proof}

By Lemma~\ref{lemma:AS:cover-structure}, the unique path covering a
leaf $u\in L(a_i)$ ends in $L(b_j)$ for some $j$. We call $j$ the
\emph{destination of} $u$. We say that item $i$ is \emph{split} if there exist two vertices in $L(a_i)$ with different destinations. Otherwise, we say that item $i$ is \emph{packed} in bin $j$. 

\begin{lemma}
    No item is split. 
    \label{claim:AS:no_split}
\end{lemma}

\begin{proof}
Fix a gadget between $a_i$ and $b_j$, of capacity $t=t_i$.
We will show that it carries either no encoding path or exactly $t_i$
encoding paths. Since there are only $t_i$ leaves at $a_i$, this will
exclude splitting item $i$ between several bins.

By Lemma~\ref{lemma:AS:cover-structure}, every path meeting the gadget
is either an encoding path from $L(a_i)$ to $L(b_j)$ or consists
entirely of its internal vertices. Encoding paths follow vertical
segments $x_r,p_r,q_r,y_r$; in particular, they never cover a vertex
of type $e_r$ or $f_r$.

We first classify paths contained in the gadget. They have length
seven. The distances of its vertex types to $(a_i,b_j)$ are
\[
\begin{array}{c|cccccccc}
 &x_r&p_r&q_r&y_r&c_r&d_r&e_r&f_r\\\hline
(d(a_i,\cdot),d(b_j,\cdot))
 &(1,4)&(2,3)&(3,2)&(4,1)&(2,3)&(3,2)&(4,3)&(3,4).
\end{array}
\]
For endpoints at distance seven, the sums of their distances to each
connection vertex are at least seven. Since all distances in the table are at most four, one endpoint
must be at distance four from $a_i$ and at least three from $b_j$.
It therefore has type $e_r$. The other endpoint must then be at
distance four from $b_j$ and at least three from $a_i$, so it has
type $f_s$.
Starting from $e_r$ and stopping at a vertex of type $f$, the neighbor
lists give short internal routes to $f_r,f_{r-1},f_{r+1}$ of lengths
one, three, and five, respectively. The only other simple internal
routes of length at most seven are
\[
\begin{split}
 &e_r,d_r,y_{r-1},q_{r-1},p_{r-1},x_{r-1},c_{r-1},f_{r-2},\\
 &e_r,d_r,c_r,f_{r-1},e_{r-1},d_{r-1},c_{r-1},f_{r-2}.
\end{split}
\]
This enumeration follows by choosing one of the two neighbors $d_r,f_r$
of $e_r$ and continuing without repeating a vertex. Thus a length-seven
geodesic must be the green path at $r-1$ or the blue path at $r$.

We can now use unique coverage to compare consecutive cells. Let
$g_r$ be one if the green path at $r$ belongs to the cover, and zero
otherwise; define $h_r$ in the same way for the blue path at $r$.
The vertex $e_r$ belongs to the green path at $r-1$ and to the blue
paths at $r$ and $r+1$. Likewise, $f_{r-1}$ belongs to the green path
at $r$ and to those same two blue paths. Neither vertex can be covered
by an encoding path. Since each is covered exactly once, we have
\[
 g_{r-1}+h_r+h_{r+1}=1,
 \qquad g_r+h_r+h_{r+1}=1.
\]
Subtracting gives $g_{r-1}=g_r$. As the indices are cyclic, all the
$g_r$ have the same value.

If this value is one, the green paths cover every vertex of $X$ in
the gadget, and unique coverage excludes any encoding path through it.
If the value is zero, no green path is used. Blue paths do not meet
$X$, so all $t_i$ vertices of $X$ must be covered by encoding paths.
Each encoding path contains exactly one such vertex, and hence there
are exactly $t_i$ encoding paths through the gadget.

Finally, choose any gadget used by a path from a leaf of $a_i$.
It carries $t_i$ encoding paths, which already account for all $t_i$
leaves of $a_i$. Thus every leaf of $a_i$ has the same destination,
and item $i$ is not split.
\end{proof}

We can use the encoding paths to construct a valid solution to \textsc{Special Unary Bin Packing}: for each $i\in [n]$, put item $i$ in the bin $j$, where $j$ is the destination of any vertex in $L(a_i)$. Exactly $t_i$ vertices of $L(b_j)$ are covered by paths from $L(a_i)$ and since all the $b$ vertices are covered, the sum of the sizes of the items placed in each bin is exactly $b$, as expected. 

This concludes the backward direction and completes the proof of Lemma \ref{lemma:AS:equivalence}.


\section{Conclusion}\label{sec:conclusion}

In this article, we studied the parameterized complexity of \SoPC. 
We proved that the problem does not admit an \XP-algorithm w.r.t.~pathwidth or treewidth, unless \P=\NP. On the contrary, \SoPC admits an XP-algorithm w.r.t.~treewidth plus diameter. However, it is \W[1]-hard w.r.t.~the combined parameters diameter and pathwidth. 

We mention two future research directions. 
First, we wonder 
whether \SoPC and \SoPP are \FPT\ or \W[1]-hard when parameterized by the solution size $k$. {Both problem are known to be in \XP~\cite{dumas2024graphs}.} Arguably, these are the most important open questions on these problems from the perspective of parameterized complexity.

In the same direction, other interesting questions are whether \SoPP or \SoPC become \FPT\ for treewidth on planar graphs. We note that this question has been raised (and is still open) for other distance-based optimization problems such as \textsc{Metric Dimension}~\cite{DBLP:journals/algorithmica/BonnetP21}.

\paragraph*{Disclosure on generative AI usage.}
This paper was initially written solely by the authors. AI tools were then used to run preliminary checks of the correctness, to obtain the arguments of some minor lemmas, to revise, and to enhance the presentation of the manuscript. 
We detail this usage for each of the sections where it was used.
\begin{description}
    \item[Section \ref{sec:SCP-tree-width-para-NP-hard}.]
The authors designed and wrote the initial reduction. Discussions with AI tools helped identify a flaw, which the tools could not resolve. The authors developed a revised approach using a different source problem; ChatGPT Astra supplied the proof of its NP-hardness (Lemma~\ref{lem:cpm-source-hard}). Further discussions revealed a missing component, leading the authors to add the ``spine'' construction. AI tools drafted the current proof from the authors' detailed construction and arguments, and the authors revised it for clarity and readability and created the figure. In addition to our own proofreading, ChatGPT Astra subsequently reviewed the reduction without identifying further errors. 

    \item[Section \ref{sec:algo-treewidth-plus-diam}.] A first version of the dynamic programming was written by the authors. AI tools were used to verify the approach and enhance the presentation.
   \item[Section \ref{sec:cover-W1h-tw-diam}.]
AI tools reviewed the reduction for possible errors and reported none. 
\end{description}

\paragraph*{Acknowledgments.}
The authors would like to thank the organizers of the 2024 Workshop of the AGC team of LIMOS at Saint-Jacques-d'Ambur, where this research was initiated.

Florent Foucaud was supported by the ANR project GRALMECO (ANR-21-CE48-0004), the French government IDEX-ISITE initiative 16-IDEX-0001 (CAP 20-25), the International Research Center ``Innovation Transportation and Production Systems'' of the I-SITE CAP 20-25, and the CNRS IRL ReLaX. 

Dibyayan Chakraborty received funding from the London Mathematical Society (Scheme 7, SC7-2526-02).

\bibliographystyle{alpha}
\bibliography{references} 

\end{document}